\documentclass{jca}
\usepackage[english]{babel}
\usepackage[utf8]{inputenc}
\usepackage{graphicx}
\usepackage{amssymb}
\usepackage{amsthm}
\usepackage{booktabs}
\usepackage{amsmath}

\let\phi=\varphi

\theoremstyle{definition}
\newtheorem{definition}{Definition}
\theoremstyle{plain}
\newtheorem{theorem}{Theorem}
\newtheorem{lemma}{Lemma}

\newenvironment{shortenum}{\begin{enumerate}
    \itemsep=0pt\parsep=0pt\parskip=0pt}{\end{enumerate}}

\def\qtext#1{\quad \text{#1} \quad}
\def\qqtext#1{\qquad \text{#1} \qquad}

\def\N{\mathbb N}
\def\Z{\mathbb Z}
\def\R{\mathbb R}

\def\P{\mathbf P}
\def\E{\mathbf E}

\def\set#1{\{\, #1 \,\}}

\def\abs#1{\mathopen|#1\mathclose|}

\def\incgraph{\includegraphics[scale=0.5]}
\def\incgraphc#1{\hfill\incgraph{#1}\hfill} 

\def\Conf{\mathrm{Conf}}
\let\from=\leftarrow

\begin{document}

\title{Density Classification Quality of the Traffic-majority Rules}%
\author{Markus Redeker}%
\institute{International Centre of Unconventional Computing,
  University of the West of England, Bristol, United
  Kingdom.\email{markus2.redeker@live.uwe.ac.uk}}
\maketitle

\begin{abstract}
  The density classification task is a famous problem in the theory of
  cellular automata. It is unsolvable for deterministic automata, but
  recently solutions for stochastic cellular automata have been found.
  One of them is a set of stochastic transition rules depending on a
  parameter $\eta$, the traffic-majority rules.

  Here I derive a simplified model for these cellular automata. It is
  valid for a subset of the initial configurations and uses random
  walks and generating functions. I compare its prediction with
  computer simulations and show that it expresses recognition quality
  and time correctly for a large range of $\eta$ values.
\end{abstract}

\section{Introduction}

This is an analysis of a specific solution by Fatès \cite{Fat`es2011a}
of the so-called \emph{density classification task} for cellular
automata.

This task is an example of a synchronisation problem. In it, several
independent agents (the \emph{cells}) must reach a common state, that
one which was initially that of the majority of them. The cells have
finite memory and restricted knowledge of each other's states, so that
none of them knows the whole system. Fatès published a solution that
relies on a randomised algorithm; while he proved that the algorithm
is successful with high probability, he provided only rough bounds for
the probability of success and the time it takes.

This article analyses the algorithm in detail for a subset of initial
conditions and provides improved bounds. It is an exploration of the
use of random walks and generating functions for the understanding of
cellular automata. The analysis turns out to be quite complex. I will
therefore only consider a subset of all possible initial
configurations and will use simplifications in the stochastic model.
We will however see that the simplifications have no great influence:
the model is valid in more cases than one would expect according to
its derivation.

I now will describe cellular automata and the density classification
task in detail and then proceed to the analysis of the solution.

\subsection{Cellular Automata}

Cellular Automata are discrete dynamical systems that are
characterised by local interactions. The following semi-formal
definition is intended to clarify the nomenclature used here.
\begin{quote}
  A \emph{cellular automaton} consists of a grid on which \emph{cells}
  are located. The grid is usually $\Z$ or $\Z^2$, sometimes higher
  dimensional. It may also be cyclic, like $\Z / n\Z$. Every cell has
  a \emph{state}, which belongs to a finite set that is the same for
  all cells. The system of all cells at a time is a
  \emph{configuration}.

  Time runs in discrete steps. At every time step a \emph{transition
    rule} is applied to every cell, and it determines the state of
  that cell one time step later. All cells follow the same transition
  rule, and the state of each cell depends only on the state of the
  cell itself and that of a finite number of its neighbours at the
  previous time step. The locations of these neighbours relatively to
  the cell is also equal for all cells.

  The transition rule can be deterministic or stochastic. In the
  latter case the rule determines for each cell and each possible
  state the probability that at the next time step the cell is in this
  state. For different cells these probabilities are stochastically
  independent.
\end{quote}

A special class of deterministic automata are the \emph{Elementary
  Cellular Automata} (ECA), popularised by Stephen Wolfram. They are
the cellular automata with radius 1 and two states and usually
designated by a number in Wolfram's numbering scheme
\cite{Wolfram1983}.

\subsection{The Density Classification Task}

Density Classification is a famous problem in the theory of cellular
automata. In it, a one-dimensional cellular automaton, with $\Z / n\Z$
as grid and with two cell states, must evolve to a configuration with
all cells in state 1 if the fraction of cells that are initially in
state 1 is higher than a given \emph{critical density} $\rho$,
otherwise to the state where all cells are in state 0. In most cases
the critical density is~$\frac12$.

Some ring sizes allow configurations in which the fraction of cells in
state~1 is exactly $\rho$; they are usually excluded from
consideration. In case of $\rho = \frac12$ this means that $n$ must be
odd.

The density classification task is important because it is a difficult
problem for cellular automata. In a cellular automaton, no cell knows
the system as a whole, and still they must cooperate according to
information that none of them has. Another difficulty is that no cell
in $\Z / n\Z$ is distinguished and could play the role of an
organiser.

In fact, the density classification task is unsolvable for
deterministic cellular automata, but that was found out only years
after the problem was posed.

There are however solution for stochastic cellular automata. One of
them is a parametric family of transition rules introduced by Fatès
\cite{Fat`es2011a} that is analysed here.

\subsection{What is a Solution?}

Before we proceed, we need to clarify what exactly a solution of the
density classification problem is. Apparently this has never been done
for the general case.

The most natural definition would be that a transition rule solves the
density classification problem if it classifies density correctly for
every ring size. I will call this a \emph{strong} solution; it is also
the interpretation used by Land and Belew \cite{Land1995}.

This definition is natural because it enforces a true computation of
the cellular automaton: It excludes the trivial case of a transition
rule of radius $\lceil \frac{n}{2} \rceil$ on a ring of $n$ cells that
solves the classification problem in one step because each cell can
see the whole ring.

In other cases, the algorithm depends on a parameter that must be
adjusted according to the ring size. An example is the family of
transition rules found by Briceño \emph{et al.}\ \cite{Briceno2013}.
They show that for every $\epsilon$ there is a transition rule which
classifies an initial configuration correctly if its density is less
that $\frac12 - \epsilon$ or greater than $\frac12 + \epsilon$. If the
ring has less than $\epsilon^{-1}$ cells, all initial configurations
have this property and are therefore classified correctly. Some
configurations on larger rings are however classified incorrectly; so
this rule is not a strong solution.

Nevertheless, it is no trivial algorithm, since no cell in such a
solution sees the whole ring. I will therefore call a family of
transition rules a \emph{weak solution} if it contains for each ring
size a rule that classifies density correctly and has a radius that is
significantly smaller than the ring size.

\subsection{Notes on the History of the Problem}

It is not completely clear who first proposed the Density
Classification Task. Apparently it became popular after Packard
\cite{Packard1988} used it as a test case for the ``edge of chaos''
hypothesis \cite{Langton1990}. He used evolutionary programming to
find cellular automata that are good classifiers; the fitness of a
transition rule was related to the fraction of initial configurations
it classifies correctly.

At that time, the best known approximate solution of the problem was a
transition rule by Gács, Kurdiumov and Levin \cite{Gacs1987}. It had
originally been proposed as a solution for a different problem, but it
also solves the classification task correctly for 97\% of the
densities, in a cellular automaton with ring size 149. In later
literature \cite{Fuk's1997,Mitchell1998,Mitchell1994,Mitchell1993} it
was used as a benchmark for other approximate solutions and as an
example for the concrete behaviour of an approximate solution of the
problem.

Among the people who evolved cellular automata for density
classification were Land and Belew. They found an automaton roughly as
successful as the solution of Gács \emph{et al.}, but
converging much faster. However, their ``difficulties in evolving even
better solutions led [them] to wonder whether such a CA actually
exists'' \cite[Introduction]{Land1995}, and they found the
impossibility proof.

After that, researchers began to modify the task, trying to find
variants that were solvable. Capcarrère, Sipper and Tomassini
\cite{Capcarr`ere1996} modified the output specification: Their
automaton is required to evolve to a configuration with the cells
alternately in states 0 and 1, except for longer blocks of 0s (if the
initial density was less than $\frac12$) or longer blocks of 1s (if
the density was lower than $\frac12$). Their solution, which
classifies all densities correctly, is the ECA rule with code number
184, the ``traffic rule''.

One can also change the number of states in the cellular automaton.
The initial and final configuration still have only cells in state 0
or 1, but in between more states are allowed. That is what Briceño
\emph{et al.}\ \cite{Briceno2013} did; their result is only a weak
solution, as we have seen.\footnote{A strong solution is however not
  excluded by Land and Belew's proof!}

Fuk\'s \cite{Fuk's1997} changed instead the algorithm and applied two
different cellular automata rules in succession. Initially, ECA 184 is
applied for $\left\lfloor \frac{n-2}2 \right\rfloor$ steps to generate
the pattern of alternating 0 and 1 cells and either longer blocks of
0s or 1s. Then ECA 232 is run; it is the ``majority rule'' in which
the next state of a cell is the state of the majority of the cells in
its 3-cell neighbourhood. With it, the cell states that are in the
minority slowly die out, and a configuration consisting completely of
0s or 1s is reached. It is a weak solution because the algorithm
depends on the ring size.

Sipper, Capcarrère and Ronald \cite{Sipper1998} changed topology and
output condition: Instead of using a ring they organised the cells as
a finite line, with a cell at the left end that stays always in state
0 and a cell at the right end that stays always in state 1. Then their
automaton evolves to a configuration where all cells in state 0 are at
the left and those in state 1 are at the right---and their number
unchanged---so the state of the cell in the middle is that of the
majority. They compared the unsolvable original version of the problem
\cite{Gacs1987} with the ``easy'' solutions---that by Capcarrère
\emph{et al.}\ \cite{Capcarr`ere1996} and their own---and concluded
that
\begin{quote}
  ``density is not an intrinsically hard problem to compute. This
  raises the general issue of identifying intrinsically hard problems
  for such local systems, and distinguishing them from those that can
  be transformed into easy problems.'' \cite[p.~902]{Sipper1998}
\end{quote}

Fuk\'s \cite{Fuk's2002} extended the task to stochastic cellular
automata. His algorithm is a randomised version of the majority rule:
\begin{quote}
  ``empty sites become occupied with a probability proportional to the
  number of occupied sites in the neighborhood, while occupied sites
  become empty with a probability proportional to the number of empty
  sites in the neighborhood.'' \cite[Introduction]{Fuk's2002}
\end{quote}
The algorithm classifies density in the sense that the probability
that the automaton reaches the configuration with all cells in state 1
grows with the fraction of 1s in the initial state. Thus it is more
probable for a configuration with density $> \frac12$ to evolve to a
configuration of 1s than to one of 0s, and vice versa. The additional
probability for a correct classification is however not very high
\cite[p.~233]{Fates2013}.

Fatès \cite{Fates2013} then introduced the so-called
``traffic-majority'' rules---a family of transition rules that look
like stochastic versions of the algorithm by Fuk\'s \cite{Fuk's1997}.
Instead of applying his two rules in sequence, at every time step and
for every cell either Rule 184 or Rule 232 is applied. The algorithm
can be ``tuned'' to arbitrarily high classification qualities, but at
a price: The higher the probability is for Rule 184, the better the
classification---but it also becomes slower. Although there is as yet
no formal proof, this is a sign that the algorithm is only a weak
solution.

The traffic-majority rules will now be investigated in more detail.

\section{Definitions}

\subsection{Notations and conventions}

This section explains the notations I will use here. Some of them are
not in common use, others are adapted to the requirements of this
work, others are nonstandard.

We will use sometimes the ``square bracket'' notation
\cite[p.~24]{Graham1989} for formulas with conditional terms.
\begin{definition}[Predicates as Numbers]
  Let $P$ be a predicate. Then $[P] = 0$ if $P$ is false; otherwise
  $[P] = 1$.
\end{definition}

The following definitions are in relatively widespread use:
\begin{definition}[Sets]
  The set of all positive integers is $\N = \{ 1, 2, 3, \dots \}$; the
  set of all non-negative integers is $\N_0 = \N \cup \{ 0 \}$.

  We write $\Z_n$ for $\Z / n\Z$.

  The cardinality of a finite set $S$ is $\# S$.

  We use the convention that the minimum of a set $S \subseteq \Z
  \cup \{ \infty \}$ becomes $\infty$ when $S$ is empty.
\end{definition}

We will use tuples and infinite sequences to arrange the elements of a
set as a larger entity. Most infinite sequences will describe the
development of an object over time; they are then indexed by the
variable $t$. Sometimes we will use a tuple of sequences. In this case
the tuple is typographically distinguished from its elements by an
arrow accent.
\begin{definition}[Sequences and Tuples]
  Let $S$ be a set and $I \subseteq \N_0$. The sequence $(x_t)_{t \in
    I}$, consisting of elements $x_t \in S$, is often referred to by
  the letter $x$. We write, ``$x = (x_t)_{t \in I}$''.

  If a finite tuple, like $(b_0, b_1, \dots, b_n)$, is referred to as
  a single entity, a letter with an arrow accent is used. So we will
  write, ``$\vec b = (b_0, b_1, \dots, b_n)$''.
\end{definition}

Random variables are functions from a sample space $\Omega$ to a set
$S$ \cite[p. 19]{Grimmett1989}. I will now introduce a notation to
write more easily about them.
\begin{definition}[Random Variables]
  \label{def:random-variables}
  Let $S$ be a set. The set of random variables with values in $S$ is
  written $S^\Omega$.
\end{definition}
Where possible, I will also adhere to the convention that random
variables are written with capital letters and their values with
lower-case letters.



Finally there is a useful abbreviation for the probability that an
event does \emph{not} happen; it will used very often.
\begin{definition}[Inverted Probability]
  \label{def:inverted-probability}
  Let $p \in \R$. Then $\bar p = 1 - p$.
\end{definition}

\subsection{Stochastic Cellular Automata}

The following definitions are already quite specific to the problem of
density classification: What I define here should strictly be called a
``stochastic cellular automaton with two states and a circular
arrangement of cells''.

The following definition describes the arrangement and the possible
states of the cells in a cellular automaton.
\begin{definition}[Configurations]
  The set $\Sigma = \{0, 1\}$ is the set of cell \emph{states} of
  the cellular automaton.

  A \emph{configuration} of a cellular automaton is a function $c
  \colon \Z_n \to \Sigma$. The set $\Z_n$ is the set of all cell
  \emph{locations}. The number $n$ is the \emph{ring size} of the
  cellular automaton. The set of configurations is $\Conf(n)$.
\end{definition}
For a configuration $c \in \Conf(n)$, the value of $c(x)$ is the state
of the cell at location $x$.

\begin{definition}[Configurations as Sequences]
  We sometimes write a configuration $c \in \Conf(n)$ as a sequence
  \begin{equation}
    c(0) c(1) c(2) \dots c(n-2) c(n-1)\,.
  \end{equation}
\end{definition}
Thus $0101010101\dots01$ is a configuration with cells alternately in
state 0 and state 1. For such sequences we will also use the language
of regular expressions. We then say that the configuration above is of
the \emph{form} $(01)^*$.

For two-states stochastic cellular automata, the transition rules can
be written in a special way, such that they generalise the transition
rule of deterministic cellular automata.
\begin{definition}[Stochastic Cellular Automaton]
  Let $r \in \N_0$. A stochastic \emph{local transition rule} of
  radius $r$ is a function
  \begin{equation}
    \label{eq:transition-rule}
    \phi \colon \Sigma^{2r+1} \to [0,1]\,.
  \end{equation}
  To $\phi$ corresponds, for every $n \in \N$, a stochastic
  \emph{global transition rule}. It is a function $\Phi \colon
  \Conf(n) \to \Conf(n)^\Omega$ with
  \begin{equation}
    \label{eq:global-transition-rule}
    \P(\Phi(c)(x) = 1) = \phi(c(x - r), \dots, c(x + r))
    \qquad
    \text{for all $x \in \Z_n$.}
  \end{equation}
  When we speak of the ``transition rule'' without further
  qualifications, we mean the local transition rule.

  A \emph{stochastic binary cellular automaton} is a pair $(\phi,
  n)$, where $\phi$ is a stochastic transition rule and $n \in \N$.

  The \emph{evolution} under $\phi$ of the initial configuration $C
  \in \Conf(n)^\Omega$ is the stochastic process $(\Phi^t(C))_{t \geq
    0} \in \left(\Conf(n)^{\N_0}\right)^\Omega$ of the iterated
  applications of $\Phi$ to $C$.
\end{definition}
If all values of the function $\phi$ are either 0 or 1, then the
cellular automaton is deterministic
and~\eqref{eq:global-transition-rule} becomes
\begin{equation}
  \label{eq:deterministic-transition}
  \Phi(c)(x) = \phi(c(x - r), \dots, c(x + r))\,.
\end{equation}

\subsection{Density Classification}

Now we can state the Density Classification Task for stochastic
cellular automata in a precise form. It is specified in terms of the
relative numbers of cells the in states 0 and 1 in the initial
configuration. For the calculations it is however often simpler to
work with the number of the cells in a given state, since it is an
integer. Therefore we define here notations for both concepts.
\begin{definition}[Density]
  Let $c \in \Conf(n)$ and $\sigma \in \Sigma$. Then
  \begin{equation}
    \label{eq:number-in-state}
    \#_\sigma(c)
    = \#\set{x \in \Z_n\colon c(x) = \sigma}
  \end{equation}
  is the number of cells in $c$ that are in state $\sigma$. The
  \emph{density} of state\/ $\sigma$ in $c$ is
  \begin{equation}
    \label{eq:density}
    \delta_\sigma(c) = \frac{\#_\sigma(c)}{n}\,.
  \end{equation}
  When we speak of ``density'' unqualified, $\delta_1$ is meant.
\end{definition}

The classification process has ended when all cells are in the same
state, that of the majority of cells in the initial configuration. But
this majority does not exist for a configuration with as many cells in
state 0 as in state 1. So we will exclude such initial configurations
and restrict $n$ to odd numbers. (This is the usual way, see e\,.g.\
\cite[p.~126]{Mitchell1998}.) So we will leave out in the next
definition the case of an initial configuration with density
$\frac12$.
\begin{definition}[End States]
  The \emph{end states} of the density classification problem are the
  configurations $e_\sigma \in \Conf(n)$ with $\sigma \in \Sigma$.
  \begin{equation}
    \label{eq:end-states}
    e_\sigma(x) = \sigma
    \qquad\text{for all $x \in \Z_n$.}
  \end{equation}
  Let $c \in \Conf(n)$ be a configuration. The \emph{end state for
    $c$} is
  \begin{equation}
    \label{eq:end-state}
    e(c) =
    \left\{
      \begin{matrix}
        e_0 & \text{if $\delta_0(c) > \frac12$}, \\
        e_1 & \text{if $\delta_1(c) > \frac12$}.
      \end{matrix}
    \right.
  \end{equation}
\end{definition}

A requirement implicit in the designation of $e_0$ and $e_1$ as ``end
states'' is that when the cellular automaton reaches these states, it
does not change again. This means that in a transition rule that
solves the Density Classification Task the configurations $e_0$ and
$e_1$ must be fixed points of the global transition rule $\Phi$. The
following definition expresses this in terms of the local transition
rule.
\begin{definition}[Admissible Rules]
  A transition rule $\phi \colon \Sigma^{2r + 1} \to [0, 1]$ is
  \emph{admissible} for the density classification problem if both
  $\phi(0, \dots, 0) = 0$ and $\phi(1, \dots, 1) = 1$.
\end{definition}

\paragraph{Classification Quality} Now we define the quality of a
transition rule $\phi$ as solution of the Density Classification Task.
We will use two measures of quality. The first one is the probability
that the cellular automaton finds the right answer, and the second one
is the time this computation takes. We begin with a definition of
these quantities for the case of a given initial configuration~$c$.
\begin{definition}[Classification Quality]
  Let $\phi$ be an admissible transition rule and let $n \in \N$ be an
  odd number. Let $c \in \Conf(n)$ be an initial configuration.

  The \emph{classification time} of $c$ for $\phi$ is the random
  variable $T(\phi, c) \in (\N_0 \cup \{\infty\})^\Omega$. It is the
  first time the evolution of $c$ reaches the state $e(c)$, or
  $\infty$ if that never happens.
  \begin{equation}
    \label{eq:classification-time}
    T(\phi, c)
    = \min\set{ t \in \N_0\colon \Phi^t(c) = e(c) }\,.
  \end{equation}
  The \emph{classification quality} of $\phi$ for $c$ is the number
  $q(\phi, c) \in [0, 1]$. It is the probability that the evolution of
  $c$ under $\phi$ reaches the state $e(c)$ at some time.
  \begin{equation}
    \label{eq:classification-quality}
    q(\phi, c) = \P( T(\phi, c) \neq \infty )\,.
  \end{equation}
\end{definition}

Next we need to specify what we mean with a random initial
configuration. Two methods are in use. The first one assigns to all
possible configurations the same probability. The second one is a
two-step process, in which the density is chosen first and then among
the configurations with that density one is selected. In both steps,
all outcomes have the same probability.

The first method assigns a high probability to densities near
$\frac12$. This is however a range where in a stochastic cellular
automaton the probability of a wrong classification is especially
high: Random fluctuations that transform a configuration with density
less than $\frac12$ to a configuration with density greater than
$\frac12$, or vice versa, can occur easily. Therefore we will use the
second method to get a more differentiated view of the cellular
automaton. The explicit choice of the initial density will also allow
us to investigate its influence on the recognition quality.

For a given ring size $n$ there are ${n \choose n_1}$ configurations
with $n_1$ cells in state 1 and $n + 1$ possible values for $n_1$, so
we define:
\begin{definition}[Initial Configurations]
  \label{def:initial-configurations}
  Let $n \in \N$ be the ring size of a cellular automaton and let $n_1
  \in \N_0$ be a number with $n_1 \leq n$.

  A \emph{random initial configuration of density $\frac{n_1}{n}$} is
  a random variable $\mathcal{I}(n, n_1) \in \Conf(n)^\Omega$ with
  probability
  \begin{align}
    \label{eq:random-initial-d}
    \P(\mathcal{I}(n, n_1) = c) &=
    \frac1{{n \choose n_1}} [\#_1(c) = n_1]
    &&\text{for all $c \in \Conf(n)$.}
  \end{align}
  A \emph{random initial configuration} is a random variable
  $\mathcal{I}_n \in \Conf(n)^\Omega$ with probability
  \begin{align}
    \label{eq:random-initial}
    \P(\mathcal{I}_n = c)
    &= \frac1{n + 1} \frac1{{n \choose \#_1(c)}}
    &&\text{for all $c \in \Conf(n)$.}
  \end{align}
\end{definition}
Now we can use these functions as arguments to the function $q$
of~\eqref{eq:classification-quality}, to get the classification qualities
for random initial configurations: For a transition rule $\phi$ and a
given ring size $n$, the \emph{classification quality} for an initial
density $d$ is $q(\phi, \mathcal{I}(n, \lfloor d n \rfloor))$, and the
\emph{global classification quality} is $q(\phi, \mathcal{I}_n)$.

\section{The Traffic-Majority Rules}

\subsection{Observations}

Here we begin with the analysis of the traffic-majority rules
\cite{Fat`es2011a}.
\begin{definition}[Traffic-Majority Rules]
  \label{def:fates-rule}
  The \emph{Traffic-Majority Rules} are a family of stochastic
  transition rules of radius 1, parameterised by a number $\eta \in
  [0, 1]$. Their transition functions are
  \begin{equation}
    \label{eq:fates-rule}
    \begin{aligned}
      \phi_\eta(0,0,0) &= 0, &
      \phi_\eta(0,0,1) &= 0, &
      \phi_\eta(1,0,0) &= \bar\eta, &
      \phi_\eta(1,0,1) &= 1,\\
      \phi_\eta(1,1,1) &= 1, &
      \phi_\eta(0,1,1) &= 1, &
      \phi_\eta(1,1,0) &= \eta, &
      \phi_\eta(0,1,0) &= 0, \\
    \end{aligned}
  \end{equation}
  The global transition function associated to a $\phi_\eta$ is
  $\Phi_\eta$.
\end{definition}
The equations in~\eqref{eq:fates-rule} are arranged in a special way
to show a symmetry in the rules: The values of $\phi_\eta(\sigma_{-1},
\sigma_0, \sigma_1)$ with $\sigma_0 = 0$ are in the top row and those
with $\sigma_0 = 1$ are in the bottom row. The we see,
\begin{equation}
  \label{eq:fates-symmetry}
  \overline{\phi(\sigma_{-1}, \sigma_0, \sigma_1)} = 
  \phi(\overline{\sigma_1}, \overline{\sigma_0}, \overline{\sigma_{-1}})
  \qquad
  \text{for all $\sigma_{-1}$, $\sigma_0$, $\sigma_1 \in \Sigma$.}
\end{equation}
This symmetry of the rule causes a similar symmetry in the
classification quality.
\begin{figure}[ht]
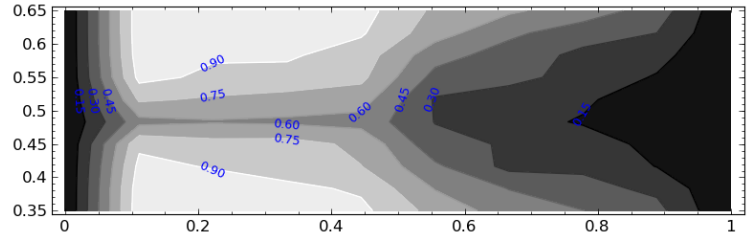

  \incgraphc{local-recognition-quality}
  \caption{Contour plot of the classification qualities $q(\phi_\eta,
    \mathcal{I}(51,n_1))$ of the traffic-majority rules. It shows the
    fraction of correct classifications for a cellular automaton of
    ring size $n =51$ that is run for 400 time steps. The horizontal
    axis is the $\eta$ parameter of the transition rule, and the
    vertical axis is the initial density $\frac{n_1}{n}$ of cells in
    state 1.}
  \label{fig:local-classification-quality}
\end{figure}
It becomes visible in Figure~\ref{fig:local-classification-quality},
which shows how the classification quality of $\phi_\eta$ depends on
$\eta$ and on the density of the initial configuration.

Another property of the traffic-majority rules that one can see in
Figure~\ref{fig:local-classification-quality} is the different
behaviour of the rule for values less and greater than $\frac12$. If
$\eta < \frac12$, the classification quality is quite good, especially
in the interval $0.15 < \eta < 0.45$, while it becomes bad very fast
for $\eta > 0.45$: For most densities of ones in an initial
configuration, the classification quality is less than $\frac12$,
which means that in most cases they are \emph{not} classified
correctly.

We can also see in Figure~\ref{fig:local-classification-quality} that
the classification quality is bad for small $\eta$. The reason for
this must be that for small $\eta$ the classification time is very
large and the classification process has not finished when the
simulation ends. This was already noted by Fatès
\cite[p.~240]{Fates2013}.

We will look into all these phenomena in more detail in the following
sections, beginning with the symmetry~\eqref{eq:fates-symmetry} of the
transition rule.

\subsection{Symmetry}

We see from Figure~\ref{fig:local-classification-quality} that the
classification quality for the initial configuration with density
$\frac{n_1}{n}$ seems to be the same as that for an initial
configuration with density $1 - \frac{n_1}{n} = \frac{n - n_1}{n}$. So
it may be that each configuration with $n_1$ cells in state 1
corresponds to a configuration with $n - n_1$ cells in state 1 that
has the same classification quality. The second configuration could be
constructed by replacing the cells in state 1 with cells in state 0
and vice versa, but~\eqref{eq:fates-symmetry} suggests that the order
of the cells must be reversed too.

This leads to the following definition, an extension of the notation
defined in Definition~\ref{def:inverted-probability} from numbers to
configurations.
\begin{definition}[Inverted Configuration]
  \label{def:inverted-config}
  Let $c \in \Conf(n)$. The \emph{inversion} of $c$ is the
  configuration $\bar c \in \Conf(n)$ with
  \begin{equation}
    \label{eq:inverted-config}
    \bar c(x) = \overline{c(n-x)}
    \qquad\text{for all $x \in \Z_n$.}
  \end{equation}
  Let $C \in \Conf(n)^\Omega$ be a random variable. The
  \emph{inversion} of $C$ is the random variable $\bar C \in
  \Conf(n)$ with
  \begin{equation}
    \label{eq:inverted-config-random}
    \P(\bar C = c) = \P(C = \bar c)
    \qquad\text{for all $c \in \Conf(n)$.}
  \end{equation}
\end{definition}

We now state the symmetry property in a general form, because we will
need it later again for another type of initial conditions.
\begin{lemma}[Symmetry]
  \label{thm:symmetry}
  Let $n \in \N$ and $c \in \Conf(n)$. Then
  \begin{equation}
    \label{eq:rq-symmetry}
    T(\phi_\eta, \bar c) = T(\phi_\eta, c)
    \qqtext{and}
    q(\phi_\eta, \bar c) = q(\phi_\eta, c)\,.
  \end{equation}
\end{lemma}
\begin{proof}
  First we prove that $\Phi_\eta(\bar c) =
  \overline{\Phi_\eta(c)}$ for all $c \in \Conf(n)$. For this, let
  $x \in \Z_n$. Then
  \begin{equation}
    \begin{aligned}[b]
      \P\left( \bigl(\overline{\Phi_\eta(c)} \bigr)(x) = 1 \right)
      \hspace{-4em} \\ 
      &= \P(\Phi_\eta(c)(n - x) = 0) \\
      &= \overline{\P(\Phi_\eta(c)(n - x) = 1)} \\
      &= \overline{\phi_\eta(
        c(n - x - 1), c(n - x), c(n - x + 1))} \\
      &= \phi_\eta \bigl(\overline{c(n - x + 1)},
      \overline{c(n - x)}, \overline{c(n - x - 1)}\bigr)
      && \quad\text{by~\eqref{eq:fates-symmetry}}\\
      &= \phi_\eta(\bar c(x - 1), \bar c(x), \bar c(x + 1)) \\
      &= P( \Phi_\eta(\bar c)(x) = 1)\,.
    \end{aligned}
  \end{equation}
  Then we can show by induction that $\overline{\Phi_\eta^t(c)} =
  \Phi_\eta^t(\bar c)$ for all $t \geq 0$. From~\eqref{eq:end-state}
  follows $e(\bar c) = \overline{e(c)}$. Therefore
  \begin{equation}
    \label{eq:endtime}
    \Phi^t(c) = e(c)
    \qtext{if and only if}
    \Phi^t(\bar c) = e(\bar c)
    \qquad\text{for all $t \geq 0$.}
  \end{equation}
  Since $T(\phi, c) = \min\{t \in \N_0: \Phi^t(c) = e(c)\}$, this
  proves that $T(\phi, c) = T(\phi, \bar c)$. Since $q(\phi, c) = \P(
  q(\phi, c) \neq \infty)$, this shows also that $q(\phi, c) = q(\phi,
  \bar c)$.
\end{proof}

We will now specialise this theorem to the random initial
configurations $\mathcal{I}(n, n_1)$. The following theorem shows that
a random configuration with density $1 - \frac{n_1}{n}$ has the same
classification quality as one with density $\frac{n_1}{n}$.
\begin{theorem}[Symmetry for $\mathcal{I}(n, n_1)$]
  \label{thm:symmetry-d}
  Let $n \in \N$, $n_1 \in \N_0$ with $n_1 \leq n$. Then
  \begin{equation}
    \label{rq-symmetry-d}
    \begin{aligned}
      T(\phi_\eta, \mathcal{I}(n, n_1))
      &= T(\phi_\eta, \mathcal{I}(n, n - n_1)), \\
      q(\phi_\eta, \mathcal{I}(n, n_1))
      &= q(\phi_\eta, \mathcal{I}(n, n - n_1))\,.
    \end{aligned}
  \end{equation}
\end{theorem}
\begin{proof}
  To apply Lemma~\ref{thm:symmetry} we need to prove that
  $\mathcal{I}(n, n - n_1) = \overline{\mathcal{I}(n, n_1)}$.

  The proof relies on the symmetry of the binomial coefficient, ${n
    \choose n - n_1} = {n \choose n_1}$ for $0 \leq n_1 \leq n$, and
  on the fact that $\#_1(c) = n - \#_1(\bar c)$ for all $c \in
  \Conf(n)$. Then we use~\eqref{eq:random-initial-d} to calculate
  \begin{equation}
    \begin{aligned}[b]
      \P (\mathcal{I}(n, n - n_1) = c)
      &= \frac{[\#_1(c) = n - n_1]}{{n \choose n - n_1}} \\
      &= \frac{[\#_1(\bar c) = n_1]}{{n \choose n_1}} \\
      &= \P(\mathcal{I}(n, n_1) = \bar c)
      = \P(\overline{\mathcal{I}(n, n_1)} = c)\,.
    \end{aligned}
  \end{equation}
  This proves that the initial configuration $\mathcal{I}(n, n -n_1)$
  has the same probability distribution as $\overline{\mathcal{I}(n,
    n_1)}$.
\end{proof}

Theorem~\ref{thm:symmetry-d} allows us to restrict our attention to
initial configurations with density less than $\frac12$. We will do
this in the remaining part of the text.

\section{Blocks of Cells in State 1}

An analysis of the behaviour of the traffic-majority rule would be
very complex if we allowed all elements of $\Conf(n)$ as initial
configurations: There is a huge variety of initial configurations,
and, since $\phi_\eta$ is non-deterministic, also a huge variety of
evolutions for a single initial configuration.

We will therefore now look at one family of initial configurations in
more detail, namely those that consist of a block of cells in state 1,
surrounded by cells in state 0. We will see that they share a common
behaviour.

\subsection{Observations}

To get a classification quality for blocks we define initial
configurations similar to those of
Definition~\ref{def:initial-configurations}.
\begin{definition}[Blocks as Initial States]
  Let $n \in \N_0$.

  Let $\ell \in \N_0$ with $\ell \leq n$. A \emph{1-block of length
    $\ell$} is a configuration $\mathcal{B}(n, \ell) \in \Conf(n)$ with
  \begin{equation}
    \label{eq:1-block}
    \mathcal{B}(n, \ell)(x) = [x \in \{0, \dots, \ell - 1\}]
    \qquad\text{for all $x \in \Z_n$.}
  \end{equation}
  A \emph{random block} is a configuration $\mathcal{B}_n \in
  \Conf(n)^\Omega$ with
  \begin{equation}
    \label{eq:random-block}
    \P(\mathcal{B}_n = \mathcal{B}(n, \ell)) = \frac1{n+1}
    \qquad\text{for all $\ell \in \N_0$ with $\ell \leq n$.}
  \end{equation}
\end{definition}
The Block $\mathcal{B}(n, \ell)$ is an initial configuration with
density $\frac{\ell}{n}$. Note that in contrast to $\mathcal{I}(n,
n_1)$, it is not a random variable.

\begin{figure}[ht]
  \incgraphc{local-block-recognition}
  \caption{Classification quality for a single block of ones under the
    traffic-majority rule. The parameters of this contour plot are the
    same as in Figure~\ref{fig:local-classification-quality}.}
  \label{fig:local-block-classification}
\end{figure}
The local classification quality for blocks
(Figure~\ref{fig:local-block-classification}) has the same general
form as that for generic initial conditions; for $\eta > \frac12$ it
is however even worse. A partial explanation for this behaviour will
be given in Section~\ref{sec:glob-class-qual}.

The symmetry property for 1-block initial conditions $\mathcal{B}(n,
\ell)$ can now be proved in the same way as that for the generic
initial conditions $\mathcal{I}(n, n_1)$.
\begin{theorem}[Symmetry for 1-Blocks]
  \label{thm:block-symmetry}
  Let $n \in \N$ and $\ell \in \N_0$ with $\ell \leq n$. Then
  \begin{equation}
    \label{eq:block-symmetry}
    \begin{aligned}
      T(\phi_\eta, \mathcal{B}(n, \ell))
      &= T(\phi_\eta, \mathcal{B}(n, n - \ell)), \\
      q(\phi_\eta, \mathcal{B}(n, \ell))
      &= q(\phi_\eta, \mathcal{B}(n, n - \ell))\,.
    \end{aligned}
  \end{equation}
\end{theorem}
\begin{proof}
  We note that $\mathcal{B}(n, n - \ell) = \overline{\mathcal{B}(n,
    \ell)}$, which can be seen directly from
  Definition~\ref{def:inverted-config} and
  Equation~\eqref{eq:1-block}. Therefore we can apply
  Lemma~\ref{thm:symmetry}.
\end{proof}

\subsection{Behaviour of the Boundaries}

\begin{figure}[ht]
  \incgraphc{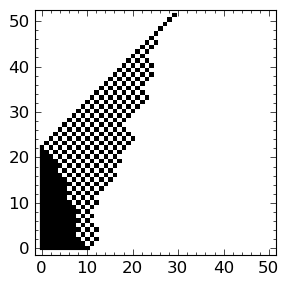}
  \caption{Phases in the development of a block. Phase~I ends at time
    step 21 and Phase~II at time step 51.}
  \label{fig:phases}
\end{figure}
Now we will look in more details into the evolution of blocks.
Figure~\ref{fig:phases} shows one example, and more samples are given
in Figure~\ref{fig:block_samples}. In all of them, the initial density
is less than $\frac12$, so the configurations are supposed to evolve
to $e_0$.

We can see from these pictures that the repertoire of configurations
that occur during the evolution of $\mathcal{B}(n, \ell)$ is limited:
They all have the general form $1^* (01)^*0^*$. We now begin with the
investigation of the behaviour of these configurations.

As we can see from Figure~\ref{fig:phases}, the evolution of a
$\mathcal{B}(n, \ell)$ seems to consist of two phases: one in which a
block of consecutive cells in state 1 is present and the other one in
which there is only the alternating pattern of 0 and 1.
\begin{definition}[Phases]
  The time $t$ in the evolution of $\mathcal{B}(n, \ell)$ in which the
  configuration $\Phi_\eta^t(\mathcal{B}(n, \ell))$ has the form
  $11^*(10)^*0^*$, is called \emph{Phase~I}.

  The time in which the configuration $\Phi_\eta^t(\mathcal{B}(n,
  \ell))$ has the form $(10)^*10^*$ or a shifted version of it, is
  \emph{Phase~II}.
\end{definition}
In both phases there must be at least one cell in state 1. The end
state $e_0$ belongs to neither of the phases; we could view it as a
third phase.

\begin{figure}[tp]
  \def\s#1{\includegraphics[scale=0.43]{block_samples/#1.png}} 
  \tabcolsep=0pt
  \center
  \begin{tabular}{ccccc}
    \toprule
    $\eta = 0$ & $\eta = 0.3$ & $\eta = 0.5$
    & $\eta = 0.7$ & $\eta = 1$ \\ \midrule
    \s{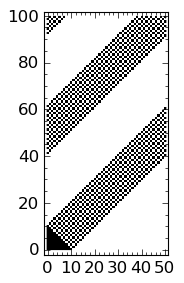} & \s{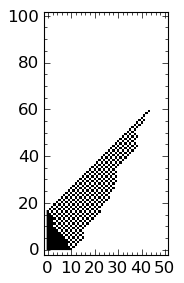} & \s{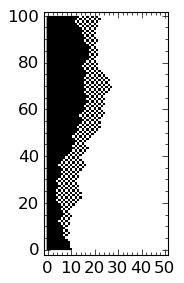} & \s{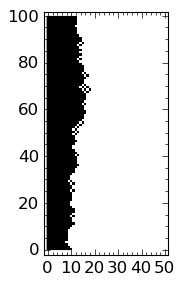} & \s{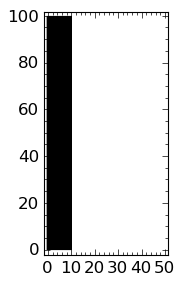} \\
    \s{0} & \s{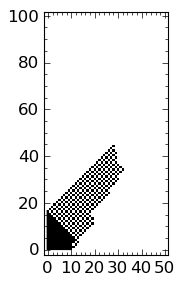} & \s{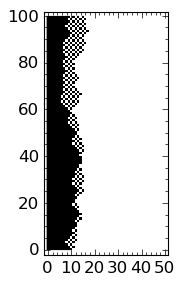} & \s{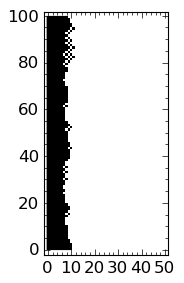} & \s{1} \\
    \s{0} & \s{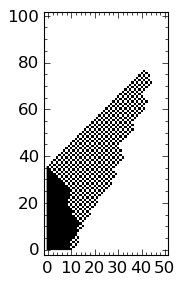} & \s{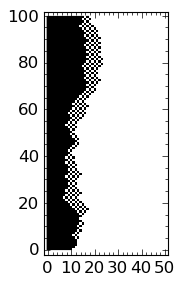} & \s{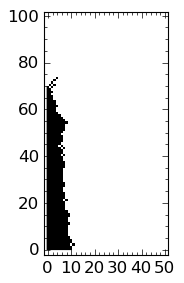} & \s{1} \\
    \s{0} & \s{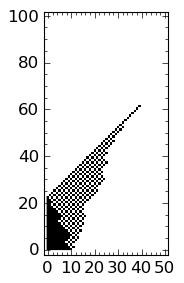} & \s{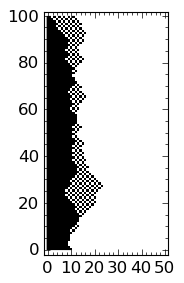} & \s{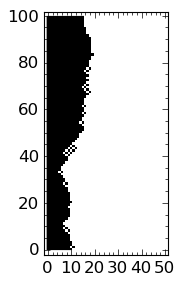} & \s{1} \\
    \bottomrule
  \end{tabular}
  \flushleft
  \caption{Evolution of a block of cells in state 1 under the
    traffic-majority rule. There are four samples for each value of
    $\eta$; in the case of $\eta = 0$ and $\eta = 1$ the rule becomes
    deterministic. The initial generation is at the bottom of each
    picture.}
  \label{fig:block_samples}
\end{figure}

The shape of such a configuration is determined by only three numbers.
We will work with them instead of cell sequences.
\begin{definition}[Boundaries]
  \label{def:boundaries}
  Let $(C_t)_{t \geq 0} = (\Phi_\eta^t(\mathcal{B}(n, \ell)))_{t
    \geq 0}$ be the evolution of a 1-block. Let $C_t$ be a cell
  sequence of the form $1^* (01)^*0^*$ in which at least one cell is
  in state 1.

  Now we define $B_{0,t}$, $B_{1,t}$ and $B_{1,t} \in (\Z_n)^\Omega$
  in such a way that
  \begin{shortenum}
  \item the region $C_t(B_{0,t}), \dots, C_t(B_{1,t})$ has the form
    $11^*$,
  \item the region $C_t(B_{1,t}), \dots, C_t(B_{2,t})$ has the form
    $1(01)^*$, and
  \item the region $C_t(B_{2,t}), \dots, C_t(B_{0,t})$ has the form
    $10^*$.
  \end{shortenum}
  These three random variables are then the \emph{boundaries} of
  $C_t$. The region between $B_{0,t}$ and $B_{1,t}$ is the
  \emph{1-block} of $C_t$ and the region between $B_{1,t}$ and
  $B_{2,t}$ is the \emph{01-sequence}.

  We will also use the tuple $\vec B_t = (B_{0,t}, B_{1,t}, B_{2,t})$
  of the boundaries at time $t$.
\end{definition}
The meaning of these three variables is shown in the following
diagram:
\begin{equation}
  \label{eq:boundaries}
  \arraycolsep=0pt
  \def\u{\uparrow}
  \begin{array}{ccccccccccccccccccc}
    011&1111111111&110&10101010101&010&0\\
    \u &&          \u &&           \u\\
     B_{0,t} &&     B_{1,t} &&       B_{2,t}
  \end{array}
\end{equation}
During Phase~I we always have $B_{0,t} \neq B_{1,t}$, and during
Phase~II there is always $B_{0,t} = B_{1,t}$.

Now we can express the effect of transition rule $\phi_\eta$ on these
boundaries. First we introduce a shorter notation for their behaviour
in a single time step:
\begin{definition}[Single Step Probability]
  Let $(b_0, b_1, b_2)$ and $(b'_0, b'_1, b'_2) \in \Z_n^3$. We define,
  \begin{equation}
    P(b'_0, b'_1, b'_2 \from b_0, b_1, b_2) =
    \P(\vec B_{t+1} = (b'_0, b'_1, b'_2) \mid
    \vec B_t = (b_0, b_1, b_2))\,.
  \end{equation}
  Often we will write $P(b'_0, b'_1, b'_2 \from \vec b)$ for $P(b'_0,
  b'_1, b'_2 \from b_0, b_1, b_2)$.

  We will also write $P(e_0 \from b_0, b_1, b_2)$ for the probability
  that the state with boundaries $\vec B_t = (b_0, b_1, b_2)$ evolves
  in the next step to $e_0$.
\end{definition}

If a configuration grows very long during Phase~I, it may ``wrap
around'', such that the right end of the 01-sequence reaches the left
end of the 1-block. In this case our model will break down. We can
however expect that this kind of behaviour has low probability when
the initial 1-block is short enough: such a block is supposed to
vanish fast and to stay small while it exists. The examples in
Figure~\ref{fig:block_samples} support this expectation.

In the following analysis we will therefore assume that the initial
1-block is so short that wrap-arounds are improbable. We will later
compare the theoretical results with actual data and find the block
lengths for which the assumption is justified.

\begin{theorem}[Single Step Behaviour]
  \label{thm:single-step}
  Let $b_2 \neq b_0$.
  \begin{enumerate}
  \item \textbf{Phase I:} If $b_0 \neq b_1$, we have the transition
    probabilities
    \begin{subequations}
      \label{eq:phase-i}
      \begin{alignat}{2}
        \label{eq:0LR}
        P(b_0, b_1-1, b_2+1 &\from b_0, b_1, b_2)
        &&= \bar\eta^2, \\
        \label{eq:0LL}
        P(b_0, b_1-1, b_2-1 &\from b_0, b_1, b_2)
        &&= \eta \bar\eta, \\
        \label{eq:0RR}
        P(b_0, b_1+1, b_2+1 &\from b_0, b_1, b_2)
        &&= \eta \bar\eta, \\
        \label{eq:0RL}
        P(b_0, b_1+1, b_2-1 &\from b_0, b_1, b_2 )
        &&= \eta^2 [b_1 \neq b_2], \\
        \label{eq:000}
        P(b_0, b_1, b_2 &\from b_0, b_1, b_2)
        &&= \eta^2 [b_1 = b_2]\,.
      \end{alignat}
    \end{subequations}
    In all other cases, $P(b'_0, b'_1, b'_2 \from b_0, b_1, b_2) = 0$.

  \item \textbf{Phase II:} If $b_0 = b_1 \neq b_2$, we have the
    transition probabilities
    \begin{subequations}
      \label{eq:phase-ii}
      \begin{alignat}{2}
        \label{eq:RRR}
        P(b_1+1, b_1+1, b_2+1 &\from b_1, b_1, b_2)
        &&= \bar\eta, \\
        \label{eq:RRL}
        P(b_1+1, b_1+1, b_2-1 &\from b_1, b_1, b_2) 
        &&= \eta [b_1 \neq b_2], \\
        \label{eq:e0trans}
        P(e_0 &\from b_1, b_1, b_2) 
        &&= \eta [b_1 = b_2]\,.
      \end{alignat}
    \end{subequations}
    In all other cases, $P(b'_1, b'_1, b'_2 \from b_1, b_1, b_2) = 0$.
  \end{enumerate}
\end{theorem}
\begin{proof}
  This is done with the help of diagrams similar
  to~\eqref{eq:boundaries}.

  \begin{enumerate}
  \item \emph{Phase I:} Since the probabilities in~\eqref{eq:0RL} and
    \eqref{eq:000} depend on whether $b_1$ and $b_2$ are equal or not,
    we must subdivide Phase~I further.
    \begin{enumerate}
    \item $b_1 \neq b_2$: This is the case where there is both a
      nontrivial 1-block and a nontrivial 01-sequence. The boundaries
      $B_1$ and $B_2$ are separate and move independently to the left
      and the right.
      \begin{equation}
        \arraycolsep=0pt
        \let\e=\eta \def\be{\bar\eta}
        \def\1{\underline{1}}
        \begin{array}{cccccccc@{\qquad}l@{\quad}l}
          0&0\11&11111111&1\1 0&1010101&0\1 0&0&
          & b_0, b_1, b_2\\
          \cmidrule{1-7}
          &  011&11111111&1\e 1&0101010&10 \be&\\
          \cmidrule{2-6}
          & 0\11&11111111&\1 01&0101010&10\1&&
          & b_0, b_1-1, b_2+1, & p=\bar\eta^2 \\
          & 0\11&11111111&\1 01&0101010&\100&&
          & b_0, b_1-1, b_2-1, & p=\bar\eta \eta \\
          & 0\11&11111111&11 \1&0101010&10\1&&
          & b_0, b_1+1, b_2+1, & p=\eta \bar\eta \\
          & 0\11&11111111&11 \1&0101010&\100&&
          & b_0, b_1+1, b_2-1, & p=\eta^2
        \end{array}
      \end{equation}
      The left side of this diagram contains partial configurations.
      Only the 1-block and the 01-sequence is displayed, the
      surrounding zeros are mostly left out.

      The first line is the initial configuration, with the locations
      of $B_{0,t}$, $B_{1,t}$ and $B_{2,t}$ underlined. The second
      line contains for each cell the probability that this cell will
      be in state 1 at time $t+1$. Most of these probabilities are 0
      or 1, which means that the state of this cell is already
      determined by its neighbourhood at time $t$, but there are two
      places, $B_{1,t}$ and $B_{2,t} + 1$, where randomness enters.

      Since each of these locations may become 0 or 1, there are four
      possible configurations at time $t + 1$. In the last four lines
      of the diagram they are listed. The second column contains the
      values for the boundaries $B_{0,t+1}$, $B_{1,t+1}$ and
      $B_{2,t+1}$ in these configurations, assuming that they had the
      values $b_0$, $b_1$ and $b_2$ in the previous time step, and the
      third column has for each configuration the probability that it
      occurs.

      So in the third line of the diagram we see the first of the
      possible configurations. In it, $B_1$ has moved one place to the
      left, $B_2$ one place to the right, and $B_0$ has stayed at 0.
      This happened because the cell at $B_{1,t}$, which could stay in
      state state 1 with probability $\eta$, has switched to state
      0, and the cell at $B_{2,t} + 1$ has switched to state 1. Both
      events are independent and have probability $\bar\eta$, so
      the probability for the whole configuration is $\bar\eta^2$.
      This proves~\eqref{eq:0LR} in case of $b_1 \neq b_2$, and the
      other lines of~\eqref{eq:phase-i} are proved similarly.

      These four cases are the only possible configurations for the
      next time step, therefore all other transitions must have
      probability 0.

    \item $b_1 = b_2$: Since $b_0=0$, this configuration is a
      $\mathcal{B}(n, b_1 + 1)$ for some $n > b_1 + 1$.
      \begin{equation}
        \arraycolsep=0pt
        \let\e=\eta \def\be{\bar\eta}
        \def\1{\underline{1}}
        \def\u{\underline{\1}}
        \begin{array}{cccccc@{\qquad}l@{\quad}l}
          0&0\11&1111111111&1\u 0&0&
          & b_0, b_1, b_1\\
          \cmidrule{1-5}
          &  011&1111111111&1\e\be\\
          \cmidrule{2-4}
          & 0\11&1111111111&\10\1&&
          & b_0, b_1-1, b_1+1, & p=\bar\eta^2 \\
          & 0\11&1111111111&\u 00&&
          & b_0, b_1-1, b_1-1, & p=\bar\eta \eta \\
          & 0\11&1111111111&11\u &&
          & b_0, b_1+1, b_1+1, & p=\eta \bar\eta \\
          & 0\11&1111111111&1\u 0&&
          & b_0, b_1, b_1,     & p=\eta^2
        \end{array}
      \end{equation}
      When $B_1$ and $B_2$ have the same value, their place is marked
      by a double underlining.
    \end{enumerate}

  \item \emph{Phase II:} This is a pure 01-sequence. The left side
    always moves to the right, and the only place where randomness
    enters is the right side. Here there are only two possibilities:
    either the length of the sequence stays the same or it shrinks.

    A special case occurs when $b_0 = b_2$. Then we have a
    configuration with exactly one cell in state 1. and it may vanish
    altogether; then the classification process has finished.

    We have therefore again two cases.
    \begin{enumerate}
    \item $b_0 \neq b_2$:
      \begin{equation}
        \arraycolsep=0pt
        \let\e=\eta \def\be{\bar\eta}
        \def\1{\underline{1}}
        \def\u{\underline{\1}}
        \begin{array}{cccccc@{\qquad}l@{\quad}l}
          0&\u 0&10101010101&0\1  0&0&
          & b_0, b_0, b_2\\
          \cmidrule{1-5}
          &  01&01010101010&10 \be&\\
          \cmidrule{2-4}
          & 0\u&01010101010&10\1&&
          & b_0+1, b_0+1, b_2+1, & p=\bar\eta \\
          & 0\u&01010101010&\100&&
          & b_0+1, b_0+1, b_2-1, & p=\eta
        \end{array}
      \end{equation}

    \item $b_0 = b_2$:
      \begin{equation}
        \arraycolsep=0pt
        \let\e=\eta \def\be{\bar\eta}
        \def\1{\underline{1}}
        \def\u{\underline{\1}}
        \def\U{\underline{\u}}
        \begin{array}{cccc@{\qquad}l@{\quad}l}
          0&\U 0&0&
          & b_0, b_0, b_0\\
          \cmidrule{1-3}
          &  0 \be&\\
          \cmidrule{2-2}
          & 0\U&&
          & b_0+1, b_0+1, b_o+1, & p=\bar\eta \\
          & 00&&
          & e_0, & p=\eta
        \end{array}
      \end{equation}
    \end{enumerate}
  \end{enumerate}
\end{proof}

The theorem allows to conclude that Phase~I and Phase~II do occur in
this order in the evolution of a $\mathcal{B}(n, \ell)$. This is
because the initial configuration belongs to Phase~I and, as long as
no wraparound occurs, the successor of a Phase~I configuration belongs
either to Phase~I or Phase~II. The change between the phases can only
occur if $b_0 + 1 = b_1$, as a special case of~\eqref{eq:0LR}
or~\eqref{eq:0LL}. The transitions of Phase~II then either preserve
the property $b_0 = b_1$ that defines this phase, or they lead to the
end configuration $e_0$.

\subsection{Global Classification Quality for 1-Blocks}
\label{sec:glob-class-qual}

We must find out how useful 1-blocks are as an estimate for the
behaviour of all initial conditions.

In this section we will look at this question in terms of the global
classification quality.
\begin{figure}[ht]
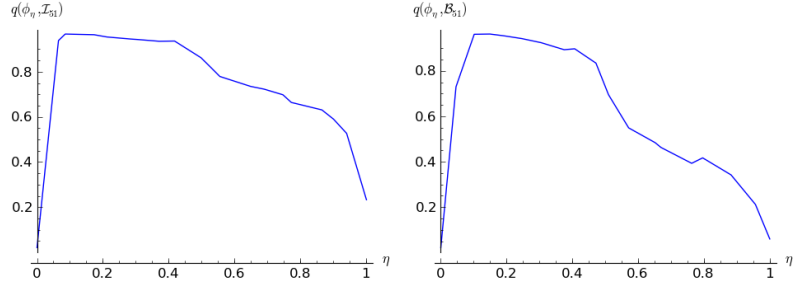

  {\centering
    \def\incgraph{\includegraphics[scale=0.43]}
    \incgraph{quality}
    \incgraph{block-quality}}
  \caption{Classification quality of the traffic-majority rule on 51
    cells for different values of $\eta$, with generic initial
    configurations (left) and with 1-blocks (right). The simulations
    were run for at most 400 time steps and repeated 1000 times for
    each value of $\eta$, each time with a random initial
    configuration.}
  \label{fig:global-recog-quality}
\end{figure}
The empirical values are shown in
Figure~\ref{fig:global-recog-quality}. We see that for $\eta <
\frac12$ the global classification quality $q(\phi_\eta,
\mathcal{I}_n)$ is very good, larger that $0.9$, and that it is almost
as good when restricted to 1-blocks. For $\eta \geq \frac12$, the
generic classification quality $q(\phi_\eta, \mathcal{I}_n)$ becomes
worse, falling almost monotonously with rising $\eta$. A qualitative
change occurs near $\eta = \frac12$, as in
Figure~\ref{fig:local-classification-quality}.

The same effect is visible for 1-block initial conditions, but it is
even stronger. For $\eta < \frac12$, the 1-blocks behave similar to
generic initial configurations, while for higher values they behave
worse. One cause for this is a theorem about \emph{archipelagos}
proved by Fatès \cite{Fat`es2011a}.

\begin{definition}[Archipelago {\cite[p.~286]{Fat`es2011a}}]
  Let $\sigma \in \Sigma$. A configuration $c \in \Conf(n)$ is a
  \emph{$\sigma$-archipelago} if all cells in state $\sigma$ are
  isolated, i.\,e.\ if there is no $x \in \Z_n$ with $c(x) = c(x + 1)
  = \sigma$.
\end{definition}
\begin{theorem}
  \label{thm:archipelagos}
  If $0 \leq \eta < 1$, then the transition rule $\phi_\eta$
  classifies all $\sigma$-archipelagos correctly.
\end{theorem}
\begin{proof}
  This is \cite[Lemma~4]{Fat`es2011a}. The condition on $\eta$ is
  not explicitly stated there, but it follows from the proof.
\end{proof}

The intersection between the sets of archipelagos and that of 1-blocks
is very small: $\mathcal{B}(n, 0) = e_0$ and $\mathcal{B}(n, n - 1)$
are 0-archipelagos, while $\mathcal{B}(n, n) = e_1$ and
$\mathcal{B}(n, 1)$ are 1-archipelagos. All other archipelagos can
occur only as random configurations $\mathcal{I}_n$.

Now we can understand why the classification quality is higher for
random initial configurations than for blocks if $\eta > \frac12$. For
these values of $\eta$, the classification quality is rather bad, as
we have seen. However, the low classification quality of $\phi_\eta$
does not affect initial configurations that are archipelagos. This
effect influences $\mathcal{I}_n$ initial configurations with a higher
probability than block initial configurations, and therefore the
classification quality for $\mathcal{I}_n$ is better.

\subsection{Deterministic Classification}

For block initial conditions and $\eta \in \{0, 1\}$ we can show that
the only 1-blocks that are classified correctly are the four
archipelago configuration described before. For these initial
configurations we can therefore actually compute the classification
quality for $\phi_\eta$:
\begin{theorem}[Block Classification in the Deterministic Cases]
  \label{thm:deterministic-classification}
  Let $n > 1$ be odd. Then
  \begin{align}
    \label{eq:deterministic-classification}
    q(\phi_0, \mathcal{B}_n) = \frac{4}{n}
    \qqtext{and}
    q(\phi_1, \mathcal{B}_n) = \frac{2}{n}\,.
  \end{align}
\end{theorem}
\begin{proof}
  We know already that $\mathcal{B}(n, 0)$, $\mathcal{B}(n, 1)$,
  $\mathcal{B}(n, n - 1)$ and $\mathcal{B}(n, n)$ are always
  classified correctly by the rule $\phi_0$. For $\phi_1$, the
  Archipelago Theorem does not apply, but $e_0$ and $e_1$ are
  trivially classified correctly. If these are the only 1-blocks that
  are classified correctly, the formulas
  in~\eqref{eq:deterministic-classification} follow. It remains
  therefore to prove that no other 1-block is classified correctly.

  We must then look at the evolution of $\mathcal{B}(n, \ell)$ for the
  values of $\ell$ for which correct classification has not been proved.
  Because of the symmetry of $\phi_\eta$
  (Theorem~\ref{thm:block-symmetry}) we only need to consider the case
  of $\ell \leq \frac{n}{2}$. In this case a correct classification
  takes place if there is a time at which the automaton is in
  configuration~$e_0$.

  The following two lemmas describe the evolution of $\mathcal{B}(n,
  \ell)$ for $\eta = 0$ and $\eta = 1$. They show that in both cases
  the configuration evolves to a state that stays essentially
  unchanged over time and is always different from $e(\mathcal{B}(n,
  \ell))$: no classification occurs, and the theorem is true.
\end{proof}

The evolution of a $\mathcal{B}(n, \ell)$ under $\eta = 0$ is shown
in the left column of Figure~\ref{fig:block_samples}.
\begin{lemma}[Block Evolution for $\eta = 0$]
  \label{thm:eta-0}
  Let $2 \leq \ell < \frac{n}{2}$. Then an evolution under $\phi_0$
  starting from $\mathcal{B}(n, \ell)$ has for $0 \leq t \leq \ell -
  1$,
  \begin{subequations}
    \begin{align}
      \label{eq:eta-0-phase-i}
      B_{0,t} &= 0,\quad B_{1,t} = \ell - 1 - t,
      & B_{2,t} &= \ell - 1 + t \\
      \intertext{and for $t \geq \ell - 1$,}
      \label{eq:eta-0-phase-ii}
      B_{0,t} &= B_{1,t} = t - \ell - 1,
      & B_{2,t} &= \ell - 1 + t\,.
    \end{align}
  \end{subequations}
  A correct classification never occurs.
\end{lemma}
\begin{proof}
  At time $t=0$ the boundaries of a $\mathcal{B}(n, \ell)$ are $B_{0,0}
  = 0$ and $B_{1,0} = B_{2,0} = \ell - 1$. We have $B_{0,0} \neq
  B_{1,0}$ because of $\ell \geq 2$ and are therefore in Phase~I.

  Among the transitions with nonzero probability in
  Theorem~\ref{thm:single-step}, only~\eqref{eq:0LR}
  and~\eqref{eq:RRR} are important here. For $\eta = 0$, they
  become
  \begin{subequations}
    \begin{alignat}{2}
      \label{eq:0LR-0}
      P(b_0, b_1-1, b_2+1 &\from \vec b)
      &&= [b_0 \neq b_1], \\
      \label{eq:RRR-0}
      P(b_1+1, b_1+1, b_2+1 &\from b_1, b_1, b_2)
      &&= [b_1 \neq b_2]\,.
    \end{alignat}
  \end{subequations}
  At the beginning, with $B_{0,0} \neq B_{1,0}$,
  transition~\eqref{eq:0LR-0} has a nonzero probability. Therefore
  $B_{0,1} = B_{0,0}$, $B_{1,1} = B_{1,0} - 1 = \ell - 2$ and $B_{2,1}
  = B_{2,0} + 1 = \ell$. So $B_1$ has moved one location to the left
  and $B_2$ has moved one location to the right. This process goes on
  for $\ell$ time steps, from $t = 0$ to $t = \ell - 1$. In all of
  them we have $B_{0,t} = 0$, $B_{1,t} = \ell - 1 - t$ and $B_{2,t} =
  \ell - 1 + t$. This proves~\eqref{eq:eta-0-phase-i}.

  At time $\ell$ we have $B_{0,\ell-1} = B_{1,\ell-1} = 0$ and
  $B_{2,\ell-1} = 2 \ell - 1$. The configuration of the automaton then
  contains a sequence $1010101\dots 01$ reaching from 0 to $2 \ell -
  1$.

  Because $\ell < \frac{n}{2}$, this structure has not yet wrapped
  around the ring of cells. In fact, since $n$ is odd, we must have $2
  \ell \leq n - 1$ and there is a gap of at least two cells in state 0
  between $B_{2,\ell-1} = 2 \ell - 1$ and $B_{0,\ell-1} = 0$.

  This becomes important for times later than $\ell - 1$. Phase~I has
  ended then, and now the process described in~\eqref{eq:RRR-0} has
  nonzero probability. This means that $B_{0,t+1} = B_{0,t} + 1$,
  $B_{1,t+1} = B_{1,t} + 1$ and $B_{2,t+1} = B_{2,t} + 1$, so the
  $01$-sequence moves with every time step one location to the right
  in the cyclic structure of the cellular automaton, a process that
  never ends. Especially the configuration $e_0$ never occurs. The
  values of $\vec B_t$ for this phase are given
  in~\eqref{eq:eta-0-phase-ii}.
\end{proof}

The evolution of a $\mathcal{B}(n, \ell)$ under $\eta = 1$ is shown
in the right column of Figure~\ref{fig:block_samples}.
\begin{lemma}[Block Evolution for $\eta = 1$]
  \label{thm:eta-1}
  Let $1 \leq \ell \leq \frac{n}{2}$. Then an evolution under $\phi_1$
  starting from $\mathcal{B}(n, \ell)$ has for all $t \geq 0$,
  \begin{equation}
    B_{0,t} = 0,\quad
    B_{1,t} = \ell - 1,\quad
    B_{2,t} = \ell - 1\,.
  \end{equation}
  A correct classification never occurs.
\end{lemma}
\begin{proof}
  At time $t=0$, the boundaries of a $\mathcal{B}(n, \ell)$ are
  $B_{0,0} = 0$ and $B_{1,0} = B_{2,0} = \ell - 1$.

  Among the processes of Theorem~\ref{thm:single-step} with nonzero
  probability, only process~\eqref{eq:000} is important for $\eta
  = 1$; it becomes
  \begin{alignat}{2}
    \label{eq:000-1}
    P(b_0, b_1, b_2 &\from \vec b) &&= [b_1 = b_2]\,.
  \end{alignat}
  This means that $B_t$ never changes: the initial configuration stays
  the same and never becomes $e_0$.
\end{proof}

For $n = 51$, Theorem~\ref{thm:deterministic-classification} implies
that $q(\phi_0, \mathcal{B}_n) \approx 0.078$ and $q(\phi_1,
\mathcal{B}_n) \approx 0.039$. This agrees with
Figure~\ref{fig:global-recog-quality}.

\section{Approximation by Random Walks}
\label{sec:appr-rand-walks}

\subsection{The Model}
\label{sec:model}

In this section we introduce a simplified model for the behaviour of
the boundaries of the 01-block, $B_{1,t}$ and $B_{2,t}$. In the model,
we assume that they are independent random variables, each performing
a random walk. We expect it to be valid if the initial configuration
is a $\mathcal{B}(n, \ell)$ for which $\ell$ is small in comparison to
$n$ and if the transition rule $\phi_\eta$ has $\eta < \frac12$.

The model is motivated by the following lemma: It shows that if
$B_{1,t}$ and $B_{2,t}$ are distinct, they indeed behave
independently.

\begin{lemma}[Boundary Independence]
  \label{thm:boundary-independence}
  In Phase I, if $B_{0,t} \neq B_{1,t} \neq B_{2,t}$, then $B_{1,t+1}$
  and $B_{2,t+1}$ are stochastically independent with the conditional
  probabilities
  \begin{subequations}
    \label{eq:boundary-separate}
    \begin{align}
      \label{eq:b1-separate}
      \P(B_{1, t+1} = x \mid B_{1,t} = b_1)
      &= [x = b_1 - 1] \bar\eta + [x = b_1 + 1] \eta, \\
      \label{eq:b2-separate}
      \P(B_{2, t+1} = x \mid B_{2,t} = b_2)
      &= [x = b_2 - 1] \eta + [x = b_2 + 1] \bar\eta\,.
    \end{align}
  \end{subequations}
  In Phase II, if $B_{1,t} \neq B_{2,t}$,
  equation~\eqref{eq:b2-separate} is valid too.
\end{lemma}
\begin{proof}
  Let $B_{i,t} = b_i$ for $i = 0, 1, 2$. Then, in Phase I,
  \begin{subequations}
    \begin{align}
      \P(B_{1,t+1} = b_1 - 1 \mid \vec B_t = \vec b)
      &= \sum_{b'_0 \in \Z} \sum_{b'_2 \in \Z}
      P(b'_0, b_1 - 1, b'_2 \from \vec b) \notag \\
      &= P(b_0, b_1 - 1, b_2 - 1 \from \vec b) \notag \\
      &\quad + P(b_0, b_1 - 1, b_2 + 1 \from \vec b) \notag  \\
      & = \bar\eta^2 + \eta \bar\eta
      = \bar\eta, \displaybreak[0] \\
      \P(B_{1,t+1} = b_1 + 1 \mid \vec B_t = \vec b)
      &= \sum_{b'_0 \in \Z} \sum_{b'_2 \in \Z}
      P(b'_0, b_1 + 1, b'_2 \from \vec b) \notag \\
      &= P(b_0, b_1 + 1, b_2 - 1 \from \vec b) \notag \\
      &\quad + P(b_0, b_1 + 1, b_2 + 1 \from \vec b) \notag  \\
      & = \eta \bar\eta + \eta^2
      = \eta,
    \end{align}
  \end{subequations}
  which proves~\eqref{eq:b1-separate}. In a similar way we
  prove~\eqref{eq:b2-separate} with the calculations
  \begin{subequations}
    \begin{align}
      \P(B_{2,t+1} = b_2 - 1 \mid \vec B_t = \vec b)
      &= P(b_0, b_1 - 1, b_2 - 1 \from \vec b) \notag \\
      &\quad + P(b_0, b_1 + 1, b_2 - 1 \from \vec b) \notag \\
      & = \bar\eta \eta + \eta^2
      = \eta, \displaybreak[0] \\
      \P(B_{2,t+1} = b_2 + 1 \mid \vec B_t = \vec b)
      &= P(b_0, b_1 - 1, b_2 + 1 \from \vec b) \notag \\
      &\quad + P(b_0, b_1 + 1, b_2 + 1 \from \vec b) \notag  \\
      & = \bar\eta^2 + \eta \bar\eta
      = \bar\eta\,.
    \end{align}
  \end{subequations}
  In Phase II, we must apply equations~\eqref{eq:phase-ii}. Here we
  have in both cases only one summand, and the computations are
  \begin{subequations}
    \begin{align}
      \P(B_{2,t+1} = b_2 - 1 &\mid \vec B_t = \vec b) \notag \\
      &= P(b_1 + 1, b_1 + 1, b_2 - 1 \from b_1, b_1, b_2)
      = \eta, \\
      \P(B_{2,t+1} = b_2 + 1 &\mid \vec B_t = \vec b) \notag \\
      &= P(b_1 + 1, b_1 + 1, b_2 + 1 \from b_1, b_1, b_2)
      = \bar\eta\,.
    \end{align}
  \end{subequations}
\end{proof}

Now a few words about the heuristic justification for the model: We
can see from the examples in the second column of
Figure~\ref{fig:block_samples} that, when the boundaries of the
01-block have separated, they stay so for a long time. This is an
experimental justification for the model. A more convincing argument
uses equations~\eqref{eq:boundary-separate}. They show that $B_{1,t}$
tends to move to the left and $B_{2,t}$ to the right if $\eta <
\frac12$. Their movements are symmetric, which means that when
$B_{1,t} = 0$ and Phase~I ends, $B_{2, t}$ must have moved
approximately $\ell$ positions to the right. So if $\ell$ is much
smaller than $\frac{n}{2}$, there is only a low probability that a
wraparound occurs, the other condition that this model could break
down. Therefore it is a good approximation for Phase~I if $\eta$ and
$\ell$ are small.

There is no danger of a wraparound in Phase~II, so the model is also
applicable to it.

\subsection{Random Walks}

We will now introduce some terminology for random walks. The random
walks used here may have an end, and this becomes part of their
definition. We will also need to specify the way that the random
walker has taken, something I have called here its \emph{path}.
\begin{definition}[Paths and Random Walks]
  A \emph{path} through $\Z$ is a pair $(s_x, x)$ with $s_x \in \N_0
  \cup \{ \infty \}$ and $x = (x_t)_{0 \leq t \leq s_x}$ a sequence of
  elements of $\Z$. The number $s_x$ is the \emph{stopping time} of
  $x$. We will usually speak of ``the path $x$'' when meaning $(s_x,
  x)$. The set of all paths is $\Pi$.

  A \emph{bounded random walk} is a random variable with values in
  $\Pi$.

  Let $p \in [0, 1]$. A \emph{$p$-walk} $(s_X, X) \in \Pi^\Omega$ is
  a bounded random walk with the property that for all $t$ with $0
  \leq t \leq s_x$,
  \begin{subequations}
    \label{eq:p-walk}
    \begin{align}
      \P(X_{t+1} = x + 1 \mid X_t = x) &= p, \\
      \P(X_{t+1} = x - 1 \mid X_t = x) &= \bar p\,.
    \end{align}
  \end{subequations}
\end{definition}

We will now translate the movement of the terms $B_{i,t}$ into the
language of bounded random walks. This is necessary because we will
later use results about random walks in $\Z$, but the boundary terms
have values in $\Z_n$. The stopping times of the new random walks will
also have a useful interpretation.

In order to simplify the computations below, we now artificially
extend the domain of $B_{2,t}$ to the first time step at which the
automaton is in state $e_0$. For this, let $t_e$ be the first time at
which the automaton is in state $e_0$. We will then set $B_{2, t_e} =
B_{2, t_e-1} - 1$. We assume thus that at the end of Phase~II, the
boundary $B_2$ takes one final step to the left.
\begin{definition}[Boundary Walks]
  \label{def:boundary-walks}
  Let $1 \leq \ell \leq \frac{n}{2}$.

  Let $(C_t)_{t \geq 0} = (\Phi_\eta^t(\mathcal{B}(n, \ell)))_{t
    \geq 0}$ be the evolution of a 1-block.

  Let $B^\Z_1$ and $B^\Z_2$ be the bounded random walks with the
  properties:
  \begin{enumerate}
  \item The stopping time $s_{B^\Z_1}$ is the first time at which
    $B_{1,t} = 0$, while $s_{B^\Z_2} = T(\phi_\eta, \mathcal{B}(n,
    \ell))$.

    If however a wraparound occurs during Phase~I, i.\,e.\ if there is
    a time $t$ with $B_{0,t} \neq B_{1,t}$ and $B_{2,t} = n - 2$, then
    $s_{B^\Z_1} = s_{B^\Z_2} = 0$.

  \item For $i = 1$, $2$ we have
    \begin{subequations}
      \begin{align}
        B^\Z_{i, 0} &= B_{i, 0}, \\
        B^\Z_{i,t} \bmod n &= B_{i, t}
        &&\text{for $0 \leq t \leq s_{B^\Z_i}$,} \\
        \label{eq:bz-diff}
        \abs{B^\Z_{i,t+1} - B^\Z_{i,t}} &\leq 1
        &&\text{for $0 \leq t < s_{B^\Z_i}$.}
      \end{align}
    \end{subequations}
  \end{enumerate}
  Then the \emph{boundary walks} for $\mathcal{B}(n, \ell)$ are the
  pair $\vec B^\Z = (B^\Z_1, B^\Z_2)$ of bounded random walks.
\end{definition}
We know from Theorem~\ref{thm:single-step} that~\eqref{eq:bz-diff} can
always be fulfilled, therefore $\vec B^\Z$ exists for every
$\mathcal{B}(n, \ell)$. The definition of the stopping times means
that $s_{B^\Z_1}$ is the first time step after Phase~I and
$s_{B^\Z_2}$ is the first time step after Phase~II.

The following lemma shows another way to characterise $s_{B^\Z_2}$. It
will be needed for the simplified model.
\begin{lemma}[Stopping Time]
  \label{thm:stopping-time}
  The stopping time $s_{B^\Z_2}$ is the smallest value of $t$ for
  which $t - B^\Z_{2,t} = s_{B^\Z_1} + 2$.
\end{lemma}
\begin{proof}
  Let $t_e = s_{B^\Z_2}$ be the end time of the classification.
  
  Let $\beta(t) = \max\{0, t - s_{B^\Z_1} \}$ be the number of time
  steps since the end of Phase~I. During Phase~I, $B_0$ stays at
  location 0, while at every time step of Phase~II, $B_0$ moves one
  location to the right. Therefore $B_{0,t} = \beta(t) \bmod n $ for
  all $t < t_e$.

  Since $B_{0,t}$ is always at the left of $B_{2,t}$, we must have
  $\beta(t) \leq B^\Z_{2,t}$ for all $t < t_e$. At time $t_e - 1$, the
  configuration consists of exactly one cell in state 1, therefore
  $\beta(t_e - 1) = B^\Z_{2,t_e-1}$. Since $\beta(t_e) = \beta(t_e -
  1) + 1$ and $B^\Z_{2,t_e} = B^\Z_{2,t_e-1} - 1$, we must have
  $\beta(t_e) - 1 = B^\Z_{2,t_e} + 1$. On the other hand, $t_e >
  s_{B^\Z_1}$ and therefore $\beta(t_e) = t - s_{B^\Z_1}$, so we must
  have $t - s_{B^\Z_1} - 1 = B^\Z_{2,t_e} + 1$, which proves the
  lemma.
\end{proof}

Now, to express the content of Lemma~\ref{thm:boundary-independence}
in the language of $p$-walks, we introduce a definition that expresses
the restrictions used in that theorem.
\begin{definition}
  \label{def:non-crossing}
  Let $N \in \Pi^2$ the set of \emph{non-crossing path pairs},
  \begin{multline}
    \label{eq:non-crossing}
    N = \set{ (x_1, x_2) \in \Pi^2 :
      s_{x_1} \leq s_{x_2}, \\
      \forall t \in \{ 1, \dots, s_{x_1} \} \colon
      0 \leq x_{1,t} < x_{2,t} \leq n - 2}
  \end{multline}
\end{definition}
Then we can say that if $\vec B^\Z \in N$, then $B^\Z_1$ has the
transition probabilities of an $\eta$-walk and $B^\Z_2$ has the
transition probabilities of an $\bar\eta$-walk, i.\,e.
\begin{subequations}
  \begin{align}
    \label{eq:bz1-separate}
    \P(B^\Z_{1, t+1} = x \mid B^\Z_{1,t} = b_1)
    &= [x = b_1 - 1] \bar\eta + [x = b_1 + 1] \eta, \\
    \label{eq:bz2-separate}
    \P(B^\Z_{2, t+1} = x \mid B^\Z_{2,t} = b_2)
    &= [x = b_2 - 1] \eta + [x = b_2 + 1] \bar\eta\,.
  \end{align}
\end{subequations}
This is Lemma~\ref{thm:boundary-independence}, translated into the
language of bounded random walks.

The simplified model then consists only of $p$-walks.
\begin{definition}[Approximation by Random Walks]
  \label{def:approximation}
  Let $\ell \in \N$ and $A_1$, $A_2 \in \Pi^\Omega$ be two bounded
  random walks that start at $\ell - 1$.

  Let $A_1$ be an $\eta$-walk that ends at the first time $t$ with
  $A_{1,t} = 0$.

  Let $A_2$ be an $\bar\eta$-walk that ends at the first time $t$
  with $t - A_{2,t} = s_{A_1} + 2$.

  Then the pair $\vec A = (A_1, A_2)$ is the \emph{approximation} for
  $\vec B$ with block length~$\ell$.
\end{definition}
In this model, $A_1$ simulates $B^\Z_1$ during Phase~I and at the
first time step of Phase~II, and $A_2$ simulates $B^\Z_2$ during both
phases and the additional time step added for
Definition~\ref{def:boundary-walks}. Therefore $s_{A_1}$ stands for
the time Phase~II begins and $s_{A_2}$ for the classification time.
The characterisation of $s_{A_2}$ follows from
Lemma~\ref{thm:stopping-time}.

To justify the definition we note that if $\vec x \in N$, then
$\P(\vec A = \vec x) = \P(\vec B^\Z = \vec x)$: every step in $x$ has
the same probability under $\vec A$ as under $\vec B^\Z$. So the
simplified model and $\vec B^\Z$ assign the same probabilities to
paths in $N$. They may however differ on the rest of the paths.

To investigate the effect of the approximation to the classification
problem we now look at $\P(s_{B_2} = t_2)$, the probability that a
classification required exactly $t_2$ steps. If we know this
probability for all $t_2$, we know $T(\phi_\eta, \mathcal{B}(n,
\ell))$ and therefore also $q(\phi_\eta, \mathcal{B}(n, \ell))$. In
the unapproximated case we can split the probability in the following
way,
\begin{align}
  \P(s_{B_2} = t_2)
  &= \P(s_{B_2} = t_2, \vec B^\Z \in N)
  + \P(s_{B_2} = t_2, \vec B^\Z \notin N) \\
\intertext{and we will do the same thing for the approximation,}
  \P(s_{A_2} = t_2)
  &= \P(s_{A_2} = t_2, \vec A \in N)
  + \P(s_{A_2} = t_2, \vec A \notin N)\,.
\end{align}
A measure for the quality of the approximation is the difference
between the two probabilities. We know that $\P(s_{B_2} = t_2, \vec
B^\Z \in N) = \P(s_{A_2} = t_2, \vec A \in N)$, because there are
again only path pairs involved with the same probability under $\vec
B^\Z$ and under $\vec A$, and therefore
\begin{multline}
  \label{eq:approx-quality}
  \abs{\P(s_{B_2} = t_2) - \P(s_{A_2} = t_2)} \\
  = \abs{\P(s_{B_2} = t_2, \vec B^\Z \notin N) -
    \P(s_{A_2} = t_2, \vec A \notin N)}\,.
\end{multline}
If this difference is small, then the approximation of $\vec B^\Z$ by
$\vec A$ is good.

We know from the proof of Lemma~\ref{thm:eta-0} that for $\eta
= 0$, all possible values of $\vec B^\Z$ are elements of $N$.
Therefore the difference~\eqref{eq:approx-quality} is 0 in this case.
Since the probability $\P(s_{B^\Z_2} = t_2, \vec B^\Z \in N)$ varies
continuously with $\eta$, we can expect that the approximation
becomes better as $\eta$ approaches 0. The probability for a
wraparound during Phase~I becomes smaller when $\ell$ is small in
comparison to $n$. So we should expect a good approximation if
$\eta \approx 0$ and $\frac{\ell}{n}$ is not too large.

What these vague conditions actually mean, we will now find out by
experiment. We will now derive the classification quality and time for
the approximation and then compare it with empirical values.

\subsection{Generating Functions}
\label{sec:generating-functions}

We will use generating functions \cite[Definition 5.1.8]{Grimmett1989}
for the random walk computations.
\begin{definition}[Generating Function]
  Let $X \in (\N_0 \cup \{ \infty \})^\Omega$ be a random variable.
  The \emph{generating\/ function} for $X$ is the function
  \begin{equation}
    \label{eq:generating-function}
    G_X(s) = \E(s^X) = \sum_{k \geq 0} \P(X = k) s^k\,.
  \end{equation}
\end{definition}

Generating functions can be used to find the expected values of random
variables.
\begin{theorem}[Properties of Generating Functions]
  Let $X \in (\N_0 \cup \{ \infty \})^\Omega$ be a random variable.
  Then
  \begin{equation}
    \label{eq:g(1)}
    G_X(1) = \P(X \neq \infty)\,.
  \end{equation}
  If $G_X(1) = 1$, then
  \begin{equation}
    \label{eq:g-expected-value}
    G_X'(1) = \E(X).
  \end{equation}
\end{theorem}
\begin{proof}
  These are Theorem 5.1.10 and formula (5.1.20) of
  \cite{Grimmett1989}.
\end{proof}
In the context of stochastic processes with discrete time, $G_X(1)$
has a more concrete interpretation. Here, ``$X = k$'' is interpreted
as, ``Event $X$ happens at time $k$'', and $X$ is an event that occurs
at most once in the duration of the process. If event~$X$ never
happens, the random variable $X$ has the value $\infty$, and therefore
$G_X(1)$ is the probability that event $X$ happens at all.

The generating functions that we will use here will have other
parameters beside $s$, but the derivatives we will need will be always
with respect to $s$; therefore the following convention will make
formulas easier to read:
\begin{definition}[Convention about Derivatives]
  In a generating function with several parameters, like $F_{-1}(p,
  s)$ below, the derivative $F_{-1}'$ is always with respect to $s$:
  We have $F_{-1}'(p, s) = \frac{d}{ds} F_{-1}(p, s)$.
\end{definition}

Another method to avoid a cumbersome notation is the use of special
names for the generating functions for $s_{A_1}$ and $s_{A_2}$.
\begin{definition}[Generating Functions for the Ends of Phase I and
  II]
  Let $\vec A$ be an approximation with block length $\ell$.

  Then $G_I(\ell, s)$ is the generating function for $s_{A_1}$ and
  $G_{II}(\ell, s)$ is the generating function for~$s_{A_2}$.
\end{definition}

\subsection{Computation of the Expected End Time}

Our next task will be to find these functions. $G_I$ is actually
well-known:
\begin{lemma}[Reaching 0]
  The generating function for the first time $t$ at which a $p$-walk
  starting from 0 reaches $X_t = 0$ again is
  \begin{equation}
    \label{eq:p-walk-back-0}
    F_0(p, s) = 1 - \sqrt{ 1 - 4 p \bar p s^2 }\,.
  \end{equation}
  The generating function for the first time a $p$-walk starting from
  1 reaches 0 is
  \begin{equation}
    \label{eq:p-walk-back-1}
    F_{-1}(p, s) = \frac{F_0(p, s)}{2 p s}\,.
  \end{equation}
  The generating function for the first time a $p$-walk starting from
  $n \in \N$ reaches 0 is
  \begin{equation}
    \label{eq:p-walk-back-n}
    F_{-n}(p, s) = F_{-1}(p, s)^n\,.
  \end{equation}
\end{lemma}
\begin{proof}
  This are Theorems 5.3.1 and 5.3.5 of \cite{Grimmett1989}. Instead
  of~\eqref{eq:p-walk-back-1} and~\eqref{eq:p-walk-back-n} they have
  formulas for $F_1$ and $F_n$; but the formulas here can be easily
  derived by noting that $F_{-n}(p, s) = F_n(\bar p, s)$.
\end{proof}

The following lemma helps us to compute properties of $F_{-1}$, which
in turn will be used for~$G_{II}$.
\begin{lemma}[Values of $F_{-1}$]
  \begin{subequations}
    \begin{align}
      \label{eq:f-1-1}
      F_{-1}(p, 1) &= \frac{\min\{p, \bar p\}}{p}, \\
      \label{eq:f-1prime-1}
      F_{-1}'(p, 1)
      &= \frac{2 \bar p}{1 - 2 \min\{p, \bar p\}}
      - \frac{\min\{p, \bar p\}}{p} \,. \\
      \intertext{If $p \leq \frac12$, this becomes}
      \label{eq:f-1-1-less}
      F_{-1}(p, 1) &= 1, \\
      \label{eq:f-1prime-1-less}
      F_{-1}'(p, 1) &= \frac1{1 - 2 p}\,.
    \end{align}
  \end{subequations}
\end{lemma}
\begin{proof}
  A useful formula for the following calculations is
  \begin{equation}
    \label{eq:sqrt-abs-min}
    \sqrt{1 - 4 p \bar p} = 1 - 2 \min\{p, \bar p\}\,.
  \end{equation}
  To prove it we note that $\sqrt{1 - 4 p \bar p} = \sqrt{1 - 4 p + 4
    p^2} = \sqrt{(1 - 2p)^2} = \abs{1 - 2p}$. If $p \leq \bar p$, then
  $p \leq \frac12$ and therefore $\abs{1 - 2p} = 1 - 2p = 1 - 2
  \min\{p, \bar p\}$. Since $\sqrt{1 - 4 p \bar p}$ is symmetric in
  $p$ and $\bar p$, the same argument is valid for $\bar p \leq p$.
  This proves~\eqref{eq:sqrt-abs-min}.

  With~\eqref{eq:sqrt-abs-min} we can compute $F_{-1}(p, 1)$:
  \begin{align}
    F_0(p, 1) &= 1 - \sqrt{1 - 4 p \bar p}
    = 2 \min\{p, \bar p\}, \\
    F_{-1}(p, 1) &= \frac{F_0(p, 1)}{2p}
    = \frac{\min\{p, \bar p\}}{p}\,.
  \end{align}
  From this we see that if $p \leq \frac12$ then $F_{-1}(p, 1) = 1$.

  For the computation of $F_{-1}'(p, 1)$ we start with the derivative
  of $F_0$.
  \begin{align}
    F_0'(p, s)
    &= - \frac1{2 \sqrt{1 - 4 p \bar p s^2}} (-4 p \bar p \cdot 2 s)
    = \frac{4 p \bar p s}{\sqrt{1 - 4 p \bar p s^2}}\,.
  \end{align}
  Next we note that the formula for the derivative of a quotient can
  be written in the form $\left(\frac{f}{g} \right)' = \frac{f'}{g} -
  \left( \frac{f}{g} \right) \frac{g'}{g}$. This leads to
  \begin{equation}
    F_{-1}'(p, s)
    = \frac{F_0'(p, s)}{2 p s}
    - F_{-1}(p, s) \frac{2p}{2ps}
    = \frac{2 \bar p}{\sqrt{1 - 4 p \bar p s^2}}
    - \frac{F_{-1}(p, s)}{s}\,.
  \end{equation}
  In case of $s = 1$ this becomes
  \begin{equation}
    F_{-1}'(p, 1)
    = \frac{2 \bar p}{1 - 2 \min\{p, \bar p\}}
    - \frac{\min\{p, \bar p\}}{p},
  \end{equation}
  and if $p \leq \frac12$,
  \begin{equation}
    F_{-1}'(p, 1) = \frac{\bar p}{1 - 2 p} - 1
    = \frac{2 - 2 p}{1 - 2 p} - 1
    = \frac1{1 - 2p}\,.
\end{equation}
\end{proof}

Since $A_1$ is an $\eta$-walk starting from $\ell - 1$ and ending
at the time $t$ at which $A_{1,t} = 0$, the generating function for
its end time $s_{A_1}$ is
\begin{equation}
  \label{eq:gen-phase-i}
  G_I(\ell, p)
  = F_{-(\ell - 1)} (\eta, s) = F_{-1}(\eta, s)^{\ell-1}
\end{equation}
The last equality follows from~\eqref{eq:p-walk-back-n}.

The stopping condition for $A_2$ is expressed in
Definition~\ref{def:approximation} in terms of the difference $t -
A_{2,t}$, not of $A_{2,t}$. To speak about paths with such a stopping
condition we now define a function $D \colon \Pi \to \Pi$, which maps
a path $x$ to the path
\begin{equation}
  \label{eq:diffpath}
  D(x) = (s_x, (t - x_t)_{0 \leq t \leq s_x})\,.
\end{equation}
So $D(X)$ and $x$ are two random walks with the same starting point
and the same number of steps. At each step, the $D(x)$ particle moves
two steps to the left if the $x$ particle moves to the right, and it
stays at its place if the $x$ particle moves to the right.

We can then say that the stochastic process $A_2$ is an
$\bar\eta$-walk that starts at $D(A_2)_0 = - \ell + 1$ (which is
because $A_{2,0} = \ell - 1$) and stops at the first time $t$ at which
$D(A_2)_t = s_{A_1} + 2$.

In the next lemma we will introduce a kind of $p$-walk in which the
value of $D$ changes exactly by 2. We can then express $A_2$ as a
sequence of these random walks. This will allow us to find the
generating functions for those $p$-walks in which the value of $D$
changes by a prescribed amount, as in $A_2$.

\begin{lemma}[Hooks]
  A \emph{hook} is a $p$-walk that ends after the first step to the
  left. Let $X$ be a hook. The generating function for its end time
  $s_X$ is
  \begin{equation}
    \label{eq:finite-walk}
    H(p, s) = \frac{\bar p s}{1 - p s}\,.
  \end{equation}
\end{lemma}
\begin{proof}
  If a hook has a stopping time $t$, it must consist of $t - 1$ steps
  to the right and then one step to the left. The steps to the right
  have probability $p$, and the step to the left has probability $\bar
  p$, therefore the probability for a walk of $t$ steps is $p^{t-1}
  \bar p$. So the generating function for its stopping time is
  $\sum_{t \geq 1} p^{t-1} \bar p s^t = \bar p s \sum_{t \geq 0} p^t
  s^t = \frac{\bar p s}{1 - p s}$.
\end{proof}

\begin{lemma}[Values of $H$]
  \begin{subequations}
    \begin{align}
      \label{eq:h-1}
      H(p, 1) &= 1, \\
      \label{eq:hprime-1}
      H'(p, 1) &= \frac1{\bar p}\,.
    \end{align}
  \end{subequations}
\end{lemma}
\begin{proof}
  Equation~\eqref{eq:h-1} is clear. For~\eqref{eq:hprime-1} we first
  note that
  \begin{equation}
    \label{eq:hprime}
    H'(p, s)
    = \frac{\bar p (1 - p s) - \bar p s (- p)}{(1 - p s)^2}
    = \frac{\bar p - \bar p p s + \bar p p s}{(1 - p s)^2}
    = \frac{\bar p}{(1 - p s)^2}\,.
  \end{equation}
  Then we can see that $H(p, 1) = \bar p / \bar p^2 = 1 / \bar p$.
\end{proof}

The following two theorems are then needed to construct $G_{II}$ from
the generating functions we have derived so far.
\begin{theorem}[Product of Generating Functions, {\cite[Theorem
    5.1.13]{Grimmett1989}}]
  \label{thm:gf-product}
  Let $X$, $Y \in \N_0^\Omega$ be two independent random variables.
  Then their sum $X + Y$ has the generating function $G_X(s) G_Y(s)$.
  \qed
\end{theorem}

\begin{theorem}[Composition of Generating Functions, {\cite[Theorem
    5.1.15]{Grimmett1989}}]
  \label{thm:gf-composition}
  Let $N$, $X \in \N^\Omega$ be two independent random variables with
  generating functions $G_N$ and $G_X$. Then the composition of $G_N$
  and $G_X$, the function $G_N \circ G_X$, is the generating function
  of the sum
  \begin{equation}
    \label{eq:composition-sum}
    X_1 + X_2 + \dots + X_N,
  \end{equation}
  where $X_1, X_2, \dots$ is a sequence of random variables that are
  independent of each other and of $N$ and all of which have the same
  distribution as $X$. \qed
\end{theorem}

In the context of random walks, the variable $X$ of the theorem may
stand for the stopping time of one random walk; then the
sum~\eqref{eq:composition-sum} is the stopping time of a sequence of
$N$ random walks, each one starting where the previous one stopped and
all of the same kind as $X$. We will now use this property to compose
the path of $A_2$ from hooks.

For the next theorem we need a lemma about the $D$ function.
\begin{lemma}[$D$ Function and Hooks]
  \label{thm:hook-decomposition}
  Let $X$ be a $p$-walk.

  Then $D(X)_{s_X} - D(X)_0$ is an even number.

  If the last step of $X$ is to the left, then $X$ is a sequence of
  $(D(X)_{s_X} - D(X)_0) / 2$ hooks.
\end{lemma}
\begin{proof}
  We have either $X_{t+1} = X_t + 1$ and $D(X)_{t+1} = D(X)_t$ or
  $X_{t+1} = X_t - 1$ and $D(X)_{t+1} = D(X)_t + 2$, which proves the
  first assertion.

  If the lasts step of $X$ is to the left, then there is a hook in $X$
  that consists of the last step of $X$ and all the steps to the right
  (possibly zero) immediately before it. The part of $X$ before that
  hook is either empty or ends with a step to the left. So $X$ is by
  induction a sequence of hooks.

  The formula for the number of hooks is true because every hook has
  one step to the left, which contributes 2 to the sum $D(X)_{s_X} -
  D(X)_0$.
\end{proof}

\begin{theorem}[Classification Time]
  \label{thm:reconition-time}
  The generating function for the approximated classification time for an
  initial configuration $\mathcal{B}_n(\ell)$ is
  \begin{equation}
    \label{eq:gen-phase-ii}
    \begin{aligned}[b]
      G_{II}(\ell, s)
      &= F_{-1} \left(\eta, \sqrt{H(\bar\eta, s)}
      \right)^{\ell-1}
      \sqrt{H(\bar\eta, s)}^{\ell+1}\,.
    \end{aligned}
  \end{equation}
\end{theorem}
\begin{proof}
  We first establish the number of hooks in $A_2$. First we note that
  the last step of $A_2$ is always to the left. Therefore $A_2$
  consists completely of hooks. We have seen before that
  $D(A_2)_{s_{A_2}} = s_{A_1} + 2$ and that $D(A_2)_0 = -\ell + 1$.
  Therefore, by Lemma~\ref{thm:hook-decomposition}, the number of
  hooks in $A_2$ is
  \begin{equation}
    \label{eq:num-hooks}
    N = \frac{s_{A_1} + \ell + 1}{2}\,.
  \end{equation}

  Next we find a generating function for $2N$: The generating function
  for $s_{A_1}$ is $G_I(\ell, s) = F_{-1}(\eta, s)^{\ell-1}$ (cf.
  \eqref{eq:gen-phase-i}), and that for $\ell + 1$ is $s^{\ell +1}$.
  Therefore, by Theorem~\ref{thm:gf-product}, the generating function
  for $2N$ is $G_{I}(\ell, s) s^{\ell + 1}$.

  This function must be an even function of the form $\sum_{k \geq 0}
  a_k s^{2k}$, because $2N$ can have only even values. It is therefore
  possible to substitute $\sqrt{s}$ for $s$ into the generating
  function for $2N$ and get a generating function for $N$, which has
  the form $\sum_{k \geq 0} a_k s^k$.

  The generating function for $N$ is then the function $G_{I}(\ell,
  \sqrt{s}) \sqrt{s}^{\ell + 1} = F_{-1}(\eta, \sqrt{s})^{\ell-1}
  \sqrt{s}^{\ell + 1}$

  To get a generating function for the sum of the stopping times for
  $N$ hooks we substitute $H(\bar\eta, s)$ for $s$ into that
  function and the result is~\eqref{eq:gen-phase-ii}, the generating
  function for $s_{A_2}$.
\end{proof}

\begin{theorem}[Classification Quality]
  If $\eta \leq \frac12$, then $G_{II}(1) = 1$.
\end{theorem}
\begin{proof}
  We have $H(\bar\eta, 1) = 1$ from~\eqref{eq:h-1}. Therefore the
  generating function for the second phase is $G_{II}(\ell, 1) =
  F_{-1}(\eta, 1)^{\ell - 1}$. Since $\eta < \frac12$, we have
  $F_{-1}(\eta, 1) = 1$ from~\eqref{eq:f-1prime-1-less} and
  therefore $G_{II}(\ell, 1) = 1$.
\end{proof}
Therefore for small $\eta$, blocks $\mathcal{B}(n, \ell)$ are
always classified correctly.

\begin{theorem}[Classification Time]
  Let $\ell \geq 2$. If $\eta \approx 0$, the expected value for
  the classification time for blocks is
  \begin{equation}
    \label{eq:block-time}
    G_{II}'(\ell, 1)
    =  \frac{\ell - \eta\ell - \eta}{\eta - 2\eta^2}\,.
  \end{equation}
\end{theorem}
\begin{proof}
  To compute this, we first define
  \begin{equation}
    f(\ell, s) = F_{-1}(\eta, s)^{\ell-1} s^{\ell+1},
  \end{equation}
  such that $G_{II}(\ell, s) = f(\ell, \sqrt{H(\bar\eta, s)})$.
  The derivative of $f$ is
  \begin{align}
    f'(\ell, s) &= (\ell - 1) F_{-1}(\ell, s)^{\ell - 2} F_{-1}'(\ell,
    s) s^{\ell+1}
    + (\ell + 1) F_{-1}(\ell, s)^{\ell - 1} s^\ell \notag \\
    &= \left((\ell - 1) F_{-1}'(\ell, s) s + (\ell + 1) F_{-1}(\ell,
      s) \right)
    f(\ell - 1, s) \,. \\
    \intertext{We know from~\eqref{eq:f-1-1-less} that
      $F_{-1}(\eta, 1) = 1$ if $\eta \leq \frac12$, therefore
      $f(\ell, 1) = 1$ and} 
    f'(\ell, 1) &= (\ell - 1) F_{-1}'(\ell, 1) + \ell + 1.
  \end{align}
  Now we can compute $G_{II}'$, first in general and then for $s = 1$.
  \begin{equation}
    \begin{aligned}[b]
      G_{II}'(\ell, s)
      &= f'(\ell, \sqrt{H(\bar\eta, s)}) 
      \frac{1}{2 \sqrt{H(\bar\eta, s)}} H'(\bar\eta, s)
    \end{aligned}
  \end{equation}
  Since we know from~\eqref{eq:hprime-1} that $H(\bar\eta, 1) = 1$,
  this term simplifies for $s = 1$ to
  \begin{equation}
    \begin{aligned}[b]
      G_{II}'(\ell, 1)
      &= \left((\ell - 1)F_{-1}'(\eta, 1) + \ell + 1 \right)
      \frac12 H'(\bar\eta, 1) \\
      &= \left((\ell - 1)\frac{1}{1 - 2 \eta}
        + \ell + 1 \right) \frac1{2\eta} \\
      &= \frac{2\ell - 2\eta\ell - 2\eta}{1-2\eta} \cdot \frac1{2\eta},
    \end{aligned}
  \end{equation}
  which proves~\eqref{eq:block-time}.
\end{proof}

\subsection{Empirical Results}
\label{sec:empirical-results}

With so many simplifications in its derivation one may doubt whether
the approximation~\eqref{eq:block-time} has any validity at all.

In the derivation of the simplified model we had argued that it would
be relatively accurate if the classification parameter $\eta$ and the
initial density $\frac{\ell}{n}$ were small. But we did not know the
ranges of $\eta$ and $\ell$ for which the approximation is good. This
question had to be resolved instead by experiment.

\paragraph{Dependence on $\eta$}

\begin{figure}[ht]
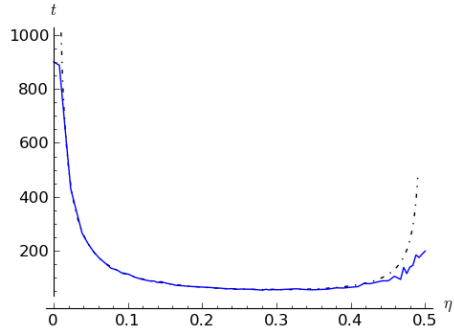

  \incgraphc{approx-epsilon}
  \caption{Classification time for $\mathcal{B}(101, 10)$ depending on
    $\eta$, empirical and as approximation. The dashed line is the
    approximated classification time~\eqref{eq:block-time} and the
    straight line represents the measured values. For every value of
    $\eta$, the simulation was run 200 times over 900 time steps
    and the resulting times then averaged. Failed classifications were
    counted as if requiring 900 time steps.}
  \label{fig:approx-eta}
\end{figure}
Figure~\ref{fig:approx-eta} compares the
approximation~\eqref{eq:block-time} of the classification time with
the empirical values for different $\eta$. The approximated
classification time was good for $0.01 < \eta < 0.4$, and the actual
classification time was always less than or equal to the predicted
time.

A curve similar to that in Figure \ref{fig:approx-eta} was found
empirically by Fatès \cite[Figure 4]{Fates2013} for arbitrary initial
configurations. Fatès also found that the classification time has its
minimum at $\eta \approx 0.3$. We can see the same minimum in
Figure~\ref{fig:approx-eta}.

There is also theoretical justification: For a given $\ell$, the
recognition time function in \eqref{eq:block-time} takes its minimum
at $\eta = \frac{2 \ell - \sqrt{2 \ell (\ell - 1)}}{2 (l + 1)}$. If
$\ell$ is large, this expression simplifies to $1 - 1 /\sqrt2 \approx
0.294$, close to the empirical values.

\paragraph{Initial Density} The connection between the initial density
and the quality of the approximation was investigated by two
experiments.

\begin{figure}[ht]
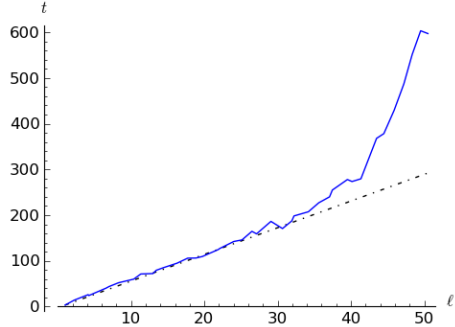

  \incgraphc{approx-ell}
  \caption{Classification time for $\mathcal{B}(101, \ell)$ depending
    on $\ell$, for $\eta = 0.3$. The other parameters are the same
    as in Figure~\ref{fig:approx-eta}.}
  \label{fig:approx-ell}
\end{figure}
In the first experiment (Figure~\ref{fig:approx-ell}), the ring size
$n$ was kept fixed and the length $\ell$ of the initial block varied.
The approximated classification time was reasonably good for $\ell$
less than about 40, corresponding to an initial density of 0.4. For
higher values of $\ell$ the actual classification time was larger than
the time predicted by the model.

\begin{figure}[ht]
  \incgraphc{approx-n-ell}
  \caption{Classification time for $\mathcal{B}(n, 10)$ depending on
    $n$, for $\eta = 0.3$. The other parameters are the same as in
    Figure~\ref{fig:approx-eta}.}
  \label{fig:approx-n-ell}
\end{figure}
In the second experiment (Figure~\ref{fig:approx-n-ell}), the length
of the initial block was kept fixed at a value of 10 and the ring size
$n$ varied. Here the quality of the approximation was good for $n$
greater than approximately 40, corresponding to a density of 0.25. For
$n < 40$, the actual classification time was always larger than the
predicted time, mostly by a large amount.

\paragraph{Ring Size}

\begin{figure}[ht]
  \incgraphc{approx-n-density}
  \caption{Classification time for $\mathcal{B}(n, \lfloor 0.2 n
    \rfloor)$ depending on $n$, for $\eta = 0.3$. The other
    parameters are the same as in Figure~\ref{fig:approx-eta}.}
  \label{fig:approx-n-density}
\end{figure}
A final experiment shows the influence of the ring size $n$ on the
quality of the approximation. The ring size does not even appear in
the approximated model, and we can ask whether that omission was
justified.

The results of the experiment are shown in
Figure~\ref{fig:approx-n-density}. Here the ring size varied and the
initial density was always 0.2. There was no visible influence of $n$
on the quality of the approximation in the data.

\section{Discussion of the Results}

\subsection{Experimental Data}

The most surprising result is certainly the good quality of the
approximation in Figure~\ref{fig:approx-eta}. It seems that the
assumption that the non-crossing path pairs are the most important
contributors to the behaviour of $\vec B^\Z$ is justified by the data,
even for relatively large $\eta$.

On the other hand, the omission of the ``early wraparound'' from the
analysis, i.\,e. the assumption that the right end of the 01-sequence
never interacts with the left end of the 1-block during Phase~I, was
not in the same measure justified by the data. Only for initial
densities less than 0.4 or even 0.25 (depending on whether one trusts
Figure~\ref{fig:approx-ell} or~\ref{fig:approx-n-ell} more), the
approximation is good. There must be a quite large space left at the
right side of the initial $\mathcal{B}(n, \ell)$ to make sure that no
significant wraparound occurs.

From Figure~\ref{fig:approx-n-density} we finally derive a different
kind of result: It shows that already for relatively modest ring
sizes, like $n \geq 50$, the actual value of $n$ is no longer
important. In cellular automata with transition rule $\phi_\eta$, the
cells then play the role of atoms, or of cells in biology, and are so
small that their number becomes unimportant.

\subsection{Conclusion}

We see that already this simplified model captures significant
properties of the traffic-majority rule. The derivation of the results
is however still quite complex, even for the simplified model. In
future, this kind of derivation could certainly be streamlined.

A puzzling fact is that the model, developed for the the evolution of
a single block, also describes the behaviour of arbitrary initial
configurations reasonably well. Apparently the traffic-majority rule
behaves as if every initial configuration were a sequence of blocks
that evolved independently of each other. A theoretical justification
of this idee would considerably simplify the analysis of this and
other cellular automata.

\paragraph{Acknowledgements} This paper was originally created as a
part of an university module under the guidance of Tim Swift. Many
thanks to him, especially for the advice to restrict the analysis of
the traffic-majority rule to the behaviour of blocks. I also want to
thank Nazim Fatès for reading a late draft and giving helpful hints.

\bibliography{../references}

\end{document}